\documentclass[journal]{IEEEtran}

\usepackage{amsmath,amssymb,amsfonts,mathtools}
\usepackage{algorithm}
\usepackage{algpseudocode}
\usepackage{array}
\usepackage{textcomp}
\usepackage{stfloats}
\usepackage{url}
\usepackage{verbatim}
\usepackage{graphicx}
\usepackage{subcaption}
\usepackage{booktabs}
\usepackage{siunitx}
\usepackage{multirow}
\usepackage{tikz}
\usetikzlibrary{positioning, arrows.meta, calc, fit, backgrounds}
\usepackage{pgfplots}
\pgfplotsset{compat=1.18}
\usepackage{amsthm}
\usepackage{bm}
\theoremstyle{definition}
\newtheorem{assumption}{Assumption}

\theoremstyle{plain}
\newtheorem{theorem}{Theorem}
\newtheorem{lemma}{Lemma}
\newtheorem{corollary}{Corollary}

\newcommand{\bA}{\mathbf{A}}

\usepackage{bm}

\newcommand{\RR}{\mathbb{R}}
\newcommand{\DS}{\mathcal{DS}}
\newcommand{\bH}{\mathbf{H}}
\newcommand{\bZ}{\mathbf{Z}}
\newcommand{\bI}{\mathbf{I}}
\newcommand{\bone}{\mathbf{1}}
\newcommand{\vecop}{\mathrm{vec}}

\usepackage[style=ieee]{biblatex}
\begin{document}

\title{Resource-efficient Semantic Coding Schemes with Manifold-constrained Hyper-connections}

\author{Jingwen Fu,~\IEEEmembership{Student Member,~IEEE,}
        Ming Xiao,~\IEEEmembership{Senior Member,~IEEE}%
\thanks{Jingwen Fu and Ming Xiao are with the School of Electrical Engineering and Computer Science (EECS), KTH Royal Institute of Technology, 11428 Stockholm, Sweden. (Corresponding author: Ming Xiao.) Email: \{jingwenf, mingx\}@kth.se.}%
}

\maketitle

\begin{abstract}
Semantic communication (SemCom) and task-oriented communication (TOC) can reduce wireless resource consumption by focusing on transmitting semantic or task-relevant information instead of raw messages. In practice, a main challenge is to make transmitting information robust to channel noise and fading while keeping it compact. Existing learning-based transceivers often improve reliability by using larger encoders or higher-dimensional channel features, which increase computation complexity and channel uses. Therefore, optimized system design needs explicit rate control to balance performance and transmitting resources e.g., bandwidth and power.
For the purpose, we propose a manifold-constrained hyper-connection (mHC) coding scheme with an entropy bottleneck (EB) for resource-efficient SemCom and TOC over wireless channels. Instead of using a single residual path of existing encoders, the proposed mHC-based semantic encoder applies multiple residual streams and constrains their interaction by doubly stochastic (DS) mixing matrices. The new structure improves representation diversity and training stability with negligible parameter and floating-point overhead. The EB quantizes the channel features and estimates the entropy-coded rate, enabling end-to-end rate--distortion/task optimization under bandwidth and transmit-power constraints.
We further show that DS-constrained stream mixing does not increase the differential entropy of the transmitted features. This implies no increase in the ideal EB coding length. Experiments on SemCom and TOC under additive white Gaussian noise (AWGN), Rayleigh fading, Rician fading, and imperfect channel state information (CSI) show that the proposed scheme improves semantic/task performance, communication robustness, and convergence stability over residual and unconstrained HC baselines, while requiring no additional channel uses.
\end{abstract}


\IEEEpeerreviewmaketitle

\section{Introduction}
\label{sec:intro}

In their seminal work, Shannon and Weaver defined three levels of communication systems: the \emph{technical} level concerning the accurate transmission of symbols, the \emph{semantic} level addressing the precision of conveyed meaning, and the \emph{effectiveness} level evaluating how the received meaning influences actions and behaviors \cite{ShannonWeaver1949}.
Conventional wireless systems have mainly addressed the technical level by optimizing metrics such as bit error rates and channel capacity, while ignoring the meaning carried by the transmitted bits.
However, in machine-to-machine communications and intelligent edge applications, the receiver often uses the received data for inference or decision-making instead of reconstructing the source messages.
This paradigm shift motivates \emph{semantic communication} (SemCom) and \emph{task-oriented communication} (TOC), where the goal is to transmit meaning or task-relevant information \cite{gunduz2023beyond_bits,qin2022semcom_principles}.

A central challenge in SemCom or TOC is to achieve high semantic representation or task completing performance under limited wireless resources~\cite{11075895}.
High performance requires diverse and robust semantic representations, especially when the channel is noisy or fading.
At the same time, practical wireless systems are constrained by bandwidth, transmit power, latency, and device computation.
Therefore, a resource-efficient SemCom/TOC system must address two coupled requirements:
(i) \emph{computational efficiency}, i.e., improving the semantic encoder without a large increase in training or inference cost, and
(ii) \emph{transmission efficiency}, i.e., reducing transmitted channel symbols while preserving the information needed for reconstruction or task inference.

Recent works have explored SemCom and TOC from different perspectives.
DeepSC shows that Transformer-based semantic transceivers can learn text transmission over noisy channels \cite{xie2021deepsc}.
Information Bottleneck (IB)-inspired TOC frameworks further introduce rate-task tradeoffs for edge inference \cite{shao2022ib_edge_inference,fu2025robust}.
These studies demonstrate the potential of end-to-end learning for SemCom and TOC.
However, several limitations remain.
First, improving performance often requires larger models or higher-dimensional channel features, which increase both computation and channel usage.
Second, many learned transceivers use fixed-dimensional continuous features or rate surrogates that do not directly correspond to an entropy-coded length.
Third, most existing semantic encoders are built on standard residual connections, which provide only a single feature pathway and may limit representation diversity in deep networks~\cite{qin2022semcom_principles}.
These limitations motivate a joint design of the semantic encoder and the rate-control mechanism.

To improve the semantic encoder, we incorporate Hyper-Connections (HC) proposed in \cite{zhu2024hyperconnections} into SemCom and TOC.
Different from a standard residual connection with one feature pathway, HC maintains multiple parallel streams and allows them to exchange information across layers.
The multiple streams provide more diverse semantic features, but unconstrained stream mixing may cause stream imbalance and unstable training.
To address the issue, manifold-constrained Hyper-Connections (mHC) are proposed by imposing a doubly stochastic (DS) constraint on the mixing matrices, which balances the contribution of different streams and improves training stability \cite{xie2025mhc}.
Therefore, mHC is suitable for improving semantic representation learning with minor extra costs.



In addition to encoder design, explicit rate control is essential for wireless SemCom and TOC.
The system should minimize the transmission rate while maintaining acceptable semantic reconstruction quality or task accuracy, which leads to a rate-distortion or rate-task optimization problem.
For this purpose, we incorporate an \emph{entropy bottleneck} (EB) that learns compressible channel features for entropy coding \cite{balle2018hyperprior}.
The EB provides a differentiable estimate of the coding rate through a learned entropy model, enabling end-to-end rate-distortion/task optimization.
Since mHC changes the feature transformation and stream mixing process in the semantic encoder, we also analyze its impact on coding efficiency.

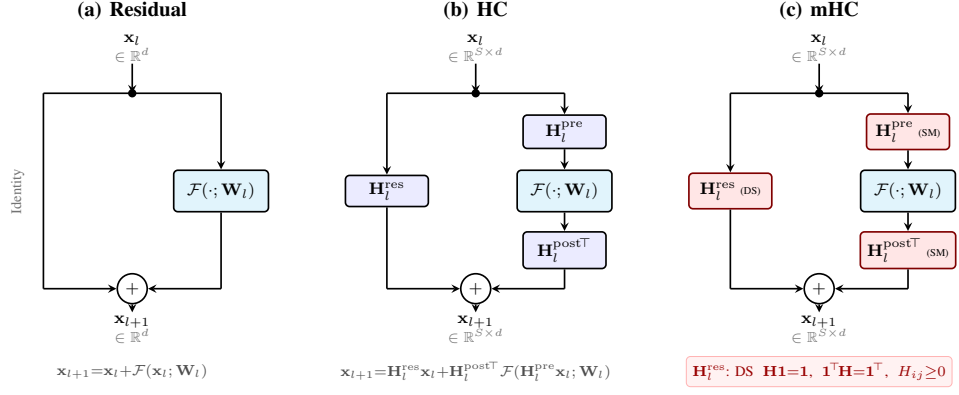
\begin{figure*}[t]
\centering
\resizebox{0.7\textwidth}{!}{%
\begin{tikzpicture}[
    >=stealth,
    font=\footnotesize,
    fblock/.style={draw, thick, rounded corners=2pt, minimum height=0.7cm,
                   minimum width=1.6cm, align=center, fill=cyan!10},
    hblock/.style={draw, thick, rounded corners=2pt, minimum height=0.5cm,
                   minimum width=1.4cm, align=center, fill=blue!8},
    dsblock/.style={draw, thick, rounded corners=2pt, minimum height=0.6cm,
                    minimum width=1.4cm, align=center, fill=red!10,
                    draw=red!50!black},
    arr/.style={->, thick, >=stealth},
    anno/.style={font=\scriptsize, text=black!55},
    ptitle/.style={font=\normalsize\bfseries, text=black},
]

\def\yt{4.6}       
\def\yinput{4.1}   
\def\ydim{3.85}    
\def\ysplit{3.2}   
\def\yhpre{2.55}   
\def\yF{1.55}      
\def\yhpost{0.55}  
\def\yjoin{-0.1}   
\def\yout{-0.65}   
\def\yodim{-0.9}   
\def\yeq{-1.5}     

\def\mainoff{1.5}  
\def\skipoff{1.5}  

\def\xa{0}

\node[ptitle] at (\xa, \yt) {(a) Residual};

\node at (\xa, \yinput) {$\mathbf{x}_l$};
\node[anno] at (\xa, \ydim) {$\in \mathbb{R}^{d}$};

\draw[arr] (\xa, {\ydim - 0.15}) -- (\xa, \ysplit);

\node[fblock] (aF) at ({\xa+\mainoff}, \yF)
    {$\mathcal{F}(\cdot;\mathbf{W}_l)$};

\node[draw, circle, inner sep=2pt, thick, fill=white, font=\footnotesize]
    (asum) at (\xa, \yjoin) {$+$};

\draw[thick] (\xa, \ysplit) -- ({\xa+\mainoff}, \ysplit);
\draw[arr] ({\xa+\mainoff}, \ysplit) -- (aF.north);
\draw[thick] (aF.south) -- ({\xa+\mainoff}, \yjoin);
\draw[arr] ({\xa+\mainoff}, \yjoin) -- (asum.east);

\draw[thick] (\xa, \ysplit) -- ({\xa-\skipoff}, \ysplit);
\draw[thick] ({\xa-\skipoff}, \ysplit) -- ({\xa-\skipoff}, \yjoin);
\draw[arr] ({\xa-\skipoff}, \yjoin) -- (asum.west);
\node[anno, rotate=90, anchor=south] at ({\xa-\skipoff-0.2}, {(\ysplit+\yjoin)/2})
    {Identity};

\fill[black] (\xa, \ysplit) circle (2pt);

\draw[arr] (asum.south) -- (\xa, {\yout + 0.15});
\node at (\xa, \yout) {$\mathbf{x}_{l+1}$};
\node[anno] at (\xa, \yodim) {$\in \mathbb{R}^{d}$};

\node[font=\scriptsize, text=black!70] at (\xa, \yeq)
    {$\mathbf{x}_{l+1}{=}\mathbf{x}_l{+}\mathcal{F}(\mathbf{x}_l;\mathbf{W}_l)$};

\def\xb{5.8}

\node[ptitle] at (\xb, \yt) {(b) HC};

\node at (\xb, \yinput) {$\mathbf{x}_l$};
\node[anno] at (\xb, \ydim) {$\in \mathbb{R}^{S \times d}$};

\draw[arr] (\xb, {\ydim - 0.15}) -- (\xb, \ysplit);

\node[hblock] (bHpre) at ({\xb+\mainoff}, \yhpre)
    {$\mathbf{H}^{\mathrm{pre}}_l$};
\node[fblock] (bF) at ({\xb+\mainoff}, \yF)
    {$\mathcal{F}(\cdot;\mathbf{W}_l)$};
\node[hblock] (bHpost) at ({\xb+\mainoff}, \yhpost)
    {$\mathbf{H}^{\mathrm{post}\!\top}_l$};

\node[hblock] (bHres) at ({\xb-\skipoff}, {(\yhpre+\yhpost)/2})
    {$\mathbf{H}^{\mathrm{res}}_l$};

\node[draw, circle, inner sep=2pt, thick, fill=white, font=\footnotesize]
    (bsum) at (\xb, \yjoin) {$+$};

\draw[thick] (\xb, \ysplit) -- ({\xb+\mainoff}, \ysplit);
\draw[arr] ({\xb+\mainoff}, \ysplit) -- (bHpre.north);
\draw[arr] (bHpre.south) -- (bF.north);
\draw[arr] (bF.south) -- (bHpost.north);
\draw[thick] (bHpost.south) -- ({\xb+\mainoff}, \yjoin);
\draw[arr] ({\xb+\mainoff}, \yjoin) -- (bsum.east);

\draw[thick] (\xb, \ysplit) -- ({\xb-\skipoff}, \ysplit);
\draw[arr] ({\xb-\skipoff}, \ysplit) -- (bHres.north);
\draw[thick] (bHres.south) -- ({\xb-\skipoff}, \yjoin);
\draw[arr] ({\xb-\skipoff}, \yjoin) -- (bsum.west);

\fill[black] (\xb, \ysplit) circle (2pt);

\draw[arr] (bsum.south) -- (\xb, {\yout + 0.15});
\node at (\xb, \yout) {$\mathbf{x}_{l+1}$};
\node[anno] at (\xb, \yodim) {$\in \mathbb{R}^{S \times d}$};

\node[font=\scriptsize, text=black!70] at (\xb, \yeq)
    {$\mathbf{x}_{l+1}{=}\mathbf{H}^{\mathrm{res}}_l\mathbf{x}_l
      {+}\mathbf{H}^{\mathrm{post}\!\top}_l
      \mathcal{F}(\mathbf{H}^{\mathrm{pre}}_l\mathbf{x}_l;\mathbf{W}_l)$};

\def\xc{11.6}

\node[ptitle] at (\xc, \yt)
    {(c) mHC};

\node at (\xc, \yinput) {$\mathbf{x}_l$};
\node[anno] at (\xc, \ydim) {$\in \mathbb{R}^{S \times d}$};

\draw[arr] (\xc, {\ydim - 0.15}) -- (\xc, \ysplit);

\node[dsblock] (cHpre) at ({\xc+\mainoff}, \yhpre)
    {$\mathbf{H}^{\mathrm{pre}}_l$\;{\tiny(SM)}};
\node[fblock] (cF) at ({\xc+\mainoff}, \yF)
    {$\mathcal{F}(\cdot;\mathbf{W}_l)$};
\node[dsblock] (cHpost) at ({\xc+\mainoff}, \yhpost)
    {$\mathbf{H}^{\mathrm{post}\!\top}_l$\;{\tiny(SM)}};

\node[dsblock] (cHres) at ({\xc-\skipoff}, {(\yhpre+\yhpost)/2})
    {$\mathbf{H}^{\mathrm{res}}_l$\;{\tiny(DS)}};

\node[draw, circle, inner sep=2pt, thick, fill=white, font=\footnotesize]
    (csum) at (\xc, \yjoin) {$+$};

\draw[thick] (\xc, \ysplit) -- ({\xc+\mainoff}, \ysplit);
\draw[arr] ({\xc+\mainoff}, \ysplit) -- (cHpre.north);
\draw[arr] (cHpre.south) -- (cF.north);
\draw[arr] (cF.south) -- (cHpost.north);
\draw[thick] (cHpost.south) -- ({\xc+\mainoff}, \yjoin);
\draw[arr] ({\xc+\mainoff}, \yjoin) -- (csum.east);

\draw[thick] (\xc, \ysplit) -- ({\xc-\skipoff}, \ysplit);
\draw[arr] ({\xc-\skipoff}, \ysplit) -- (cHres.north);
\draw[thick] (cHres.south) -- ({\xc-\skipoff}, \yjoin);
\draw[arr] ({\xc-\skipoff}, \yjoin) -- (csum.west);

\fill[black] (\xc, \ysplit) circle (2pt);

\draw[arr] (csum.south) -- (\xc, {\yout + 0.15});
\node at (\xc, \yout) {$\mathbf{x}_{l+1}$};
\node[anno] at (\xc, \yodim) {$\in \mathbb{R}^{S \times d}$};

\node[draw, rounded corners=2pt, fill=red!5, inner sep=3pt,
      font=\scriptsize, text=red!60!black, align=center, draw=red!30]
    at (\xc, \yeq)
    {$\mathbf{H}^{\mathrm{res}}_l$: DS\; $\mathbf{H}\mathbf{1}{=}\mathbf{1}$,\;
     $\mathbf{1}^{\!\top}\!\mathbf{H}{=}\mathbf{1}^{\!\top}$,\;
     $H_{ij}{\ge}0$};
 
\end{tikzpicture}
}
\caption{Three connection paradigms: (a)~standard residual connection, (b)~HC, and (c)~mHC. mHC constrains the residual mixing matrix $\mathbf{H}^{\mathrm{res}}_l$ to the DS manifold and normalizes $\mathbf{H}^{\mathrm{pre}}_l$, $\mathbf{H}^{\mathrm{post}}_l$ via softmax (SM) operation.}
\label{fig:three_arc}
\end{figure*}

Based on these insights, we propose a resource-efficient learning-based method termed the \textit{mHC coding scheme}, which can be instantiated as \textit{mHC-SemCom} and \textit{mHC-TOC} with different task heads and loss functions.
The proposed scheme combines mHC-based semantic encoding with EB-based rate control to improve semantic/task performance, transmission efficiency, and training stability over wireless channels.
The main contributions of this paper are summarized as follows:
\begin{enumerate}
\item To the best of our knowledge, we are the first to apply HC and mHC to SemCom/TOC transceivers for resource-efficient wireless communication.
      The DS-constrained multi-stream encoder improves representation diversity and training stability with limited additional overhead.

\item We integrate an EB-based entropy coder into the transmitter to provide an estimate of the coding rate.
      This enables explicit rate-distortion/task optimization for bandwidth- and power-limited wireless transmission.

\item We show that the proposed multi-stream encoder improves representation learning without increasing the communication load.
      In particular, DS-constrained mHC mixing does not increase the ideal entropy-coding length, and the proposed architecture increases only negligible computational, parameter, and transmission overhead.

\item We conduct extensive experiments under additive white Gaussian noise (AWGN), Rayleigh fading, and Rician fading channels, including settings with imperfect channel state information (CSI).
      The results show that the proposed methods outperform strong baselines in both SemCom and TOC, while improving communication robustness, resource efficiency, and training stability.
\end{enumerate}

The remainder of the paper is organized as follows. Section~\ref{sec:related} reviews related works. Section~\ref{sec:system_model} presents the system model and problem formulation. Section~\ref{sec:proposed} describes the proposed \textit{mHC coding scheme} and analyzes its resource efficiency. Section~\ref{sec:theory} provides theoretical analysis. Section~\ref{sec:experiments} reports numerical results, and Section~\ref{sec:conclusion} concludes the paper.

\textit{Notation:}
Throughout this paper, bold lowercase letters denote vectors and bold uppercase letters denote matrices.
$(\cdot)^\top$ denotes the transpose $\mathbf{I}$ means the identity matrix, and $\mathbf{1}$ denotes the all-one vector.
$\|\cdot\|_p$ denotes the $\ell_p$ norm (or spectral norm for matrices when $p=2$).
$\mathbb{R}$ denotes the set of real numbers, and $\mathbb{E}[\cdot]$ denotes expectation.


\section{Related Works}
\label{sec:related}

\subsection{Resource-Efficient SemCom and TOC}
SemCom and TOC have emerged as promising communication paradigms beyond conventional bit-level reliable transmission by focusing on meaning preservation and task completing.
Deep semantic communications (DeepSC) leverage deep neural architectures (e.g., Transformer-based semantic transceivers) to optimize semantic-level objectives over noisy channels \cite{xie2021deepsc}.
Meanwhile, end-to-end learned joint source-channel coding has demonstrated strong robustness for wireless image transmission, where deep joint source-channel coding (DeepJSCC) maps source samples to channel symbols without explicit separation of source and channel codes \cite{bourtsoulatze2019deepjscc}.
These works highlight the effectiveness of learning-based transceivers for semantic fidelity and robustness. However, the performance gains often come with the price of increased model complexity and communication loads \cite{qin2022semcom_principles}.
Meanwhile, for TOC, a key objective is to transmit only task-relevant information under stringent communication and inference constraints.
IB-based formulations have been widely adopted to explicitly characterize the rate-task tradeoff, enabling end-to-end learning and optimization of task representations \cite{shao2022ib_edge_inference}.
 Despite of these advances, two challenges remain underexplored:
\emph{(i)} how to maintain strong semantic/task performance while controlling the transmission rate under limited bandwidth, and
\emph{(ii)} how to train deep semantic transceivers stably without incurring high computational cost.

\subsection{Rate Control and Channel Codes in Learned Encoder/Decoder}
Most learned SemCom/TOC codes map semantic messages to continuous channel symbols with a fixed rate (in e.g., DeepSC/DeepJSCC)~\cite{xie2021deepsc,bourtsoulatze2019deepjscc}. Though  straightforward, such methods do not provide explicit code rates or length, and are hard to be adapted to varying channels.
To deal with time-varying channels and bandwidth constraints, prior work studies adaptive-rate transmission by varying channel symbols \cite{yang2021adaptive_rate_jscc}. Feedback-based refinement is also explored to progressively improve reconstruction quality \cite{kurka2019deepjsccf}.
Meanwhile VIB-based channel coding with a KL penalty as a rate surrogate is also proposed \cite{shao2022ib_edge_inference}. However, KL penalty does not evaluate an actual code length and may lead to posterior collapse.
EB method instead quantizes features and learns an entropy model. It yields a coding-aligned and differentiable rate estimate, enabling explicit rate control via rate--distortion optimization \cite{balle2018hyperprior}.
Thus, we will adopt EB-based rate optimization and further improve representation richness and training stability using mHC.

\subsection{Dynamic Neural Networks and Hyper-Connections}
Dynamic neural networks adapt their computation or parameters during inference to trade accuracy for computational cost \cite{han2022dynamic_survey}.
They are commonly categorized into dynamic width (activating a variable number of channels) \cite{yu2019slimmable},
dynamic depth (executing a variable number of layers/steps via early-exit or adaptive halting) \cite{graves2016act}, dynamic-parameter methods  (generating or modulating weights conditioned on inputs or layer indices) \cite{ha2017hypernetworks}, and dynamic scheduling (adaptively allocating computation across iterations) \cite{pu2025art}.
HC can be viewed as a dynamic-parameter mechanism specialized to residual networks, introducing multiple residual streams with lightweight learned mixing to enrich inter-layer information flow at negligible parameter overheads \cite{zhu2024hyperconnections}.
mHC further imposes structured constraints on the mixing to improve training stability and mitigate stream imbalance or collapse  \cite{xie2025mhc}.
Motivated by these facts, we will incorporate mHC into our proposed mHC-coded SemCom/TOC to enhance representation richness and optimization stability under stringent wireless resource constraints.

\section{System Model and Problem Formulation}
\label{sec:system_model}

\begin{figure*}[t]
\centering
\resizebox{0.98\textwidth}{!}{%
\begin{tikzpicture}[
    >=latex,
    font=\footnotesize,
    block/.style={draw, rounded corners=2pt, minimum height=0.8cm,
                  minimum width=1.5cm, align=center, fill=#1},
    block/.default=blue!8,
    bigbox/.style={draw, dashed, rounded corners=4pt, inner sep=8pt,
                   fill=#1},
    bigbox/.default=white,
    arr/.style={->, thick, >=stealth},
    darr/.style={->, thick, >=stealth, dashed},
    anno/.style={font=\scriptsize, text=black!60},
    sectitle/.style={font=\footnotesize\bfseries, text=black!80},
]

\node[bigbox=blue!3, minimum width=10.8cm, minimum height=3.8cm]
     (txbox) at (5.0, 0) {};
\node[sectitle, anchor=north] at (txbox.north) {Transmitter};

\node[font=\normalsize] (input) at (-1.2, 0) {$\mathbf{x}$};

\node[block=orange!12, minimum width=2.4cm, minimum height=1.6cm,
      draw=orange!50!black, thick]
     (enc) at (1.5, 0)
     {Semantic Encoder\\[1pt]$f_\theta$ (with mHC)};

\node[block=blue!12, minimum width=1.5cm] (chenc) at (4.5, 0)
     {Channel\\[-1pt]Encoder $g_\phi$};

\node[block=gray!10] (pnorm) at (6.4, 0) {$\mathcal{N}(\cdot)$};

\node[bigbox=green!5, minimum width=2.5cm, minimum height=2.8cm]
     (ebbox) at (8.5, -0.3) {};
\node[sectitle, anchor=north] at (ebbox.north) {Entropy Bottleneck};

\node[block=green!15, minimum width=1.3cm] (quant) at (8.5, 0) {$Q(\cdot)$};
\node[block=green!15, minimum width=1.9cm] (emodel)
     at (8.5, -1.2) {Entropy Model\\[-1pt]$p_\psi(\cdot)$};

\node[font=\small, text=green!40!black] (rate)
     at (8.5, -2.7) {$R(\hat{\mathbf{z}})$};
\draw[darr, green!50!black] (emodel.south) -- (rate.north);

\draw[arr] (input.east) -- (enc.west);
\draw[arr] (enc.east) -- (chenc.west);
\node[anno, anchor=south] at (3.0, 0.15) {$\mathbf{u}$};
\draw[arr] (chenc.east) -- (pnorm.west);
\node[anno, anchor=south] at (5.45, 0.15) {$\mathbf{z}$};
\draw[arr] (pnorm.east) -- (quant.west);
\node[anno, anchor=south] at (7.45, 0.15) {$\bar{\mathbf{z}}$};
\draw[darr, green!50!black] (quant.south) -- (emodel.north);
\node[anno, anchor=west] at (8.75, -0.6) {$\hat{\mathbf{z}}$};

\node[draw, thick, fill=yellow!18, rounded corners=3pt,
      minimum height=1.8cm, minimum width=1.7cm, align=center]
     (channel) at (11.2, 0)
     {Wireless\\[-1pt]Channel\\[3pt]
      {\scriptsize$\mathbf{y}{=}h\hat{\mathbf{z}}{+}\mathbf{n}$}};

\node[anno] at (11.2, 1.55)
     {$h{=}1$ (AWGN)\;\;/\;\; $h{\sim}\mathcal{CN}(0,1)$ (Rayleigh)};
\node[anno] at (11.2, -1.4)
     {$\mathbf{n}{\sim}\mathcal{N}(\mathbf{0},\sigma^2\mathbf{I})$};

\draw[arr] (quant.east) -- ++(0.35,0) |- (channel.west);
\node[anno, anchor=south] at (9.9, 0.15) {$\hat{\mathbf{z}}$};

\node[bigbox=purple!4, minimum width=6.0cm, minimum height=3.8cm]
     (rxbox) at (16.5, 0) {};
\node[sectitle, anchor=north] at (rxbox.north) {Receiver};

\node[block=purple!12, minimum width=1.5cm] (chdec) at (14.0, 0)
     {Channel\\[-1pt]Decoder $r_{\phi'}$};

\draw[arr] (channel.east) -- (chdec.west);
\node[anno, anchor=south] at (12.6, 0.15) {$\mathbf{y}$};

\node[block=teal!15, minimum width=1.7cm] (semhead)
     at (16.8, 0.8)
     {Semantic Head\\[-1pt]$t^{\mathrm{sem}}_\omega$};
\node[font=\small] (sout) at (19.0, 0.8) {$\hat{\mathbf{s}}$};

\node[block=violet!15, minimum width=1.7cm] (taskhead)
     at (16.8, -0.8)
     {Task Head\\[-1pt]$t^{\mathrm{task}}_{\omega'}$};
\node[font=\small] (tout) at (19.0, -0.8) {$\hat{\mathbf{s}}$};

\node[anno, anchor=west] at (19.4, 0.8)
     {\textit{mHC-Codec (Sem.)}};
\node[anno, anchor=west] at (19.4, -0.8)
     {\textit{mHC-Codec (TOC)}};

\draw[arr] (chdec.east) -- ++(0.5,0) coordinate (split);
\fill[black] (split) circle (1.5pt);
\draw[arr] (split) |- (semhead.west);
\draw[arr] (split) |- (taskhead.west);
\draw[arr] (semhead.east) -- (sout.west);
\draw[arr] (taskhead.east) -- (tout.west);

\node[anno, anchor=south] at (15.0, 0.15) {$\hat{\mathbf{u}}$};

\node[draw, thick, rounded corners=3pt, fill=red!8,
      minimum width=5.5cm, minimum height=0.65cm, align=center]
     (loss) at (13.0, -3.3)
     {$\mathcal{L} = \mathcal{L}_{\mathrm{task}}(\hat{\mathbf{s}},\mathbf{s})
       + \lambda\, R(\hat{\mathbf{z}})$};

\draw[darr, red!60] (rate.south) -- ++(0,-0.15) -| (loss.west);
\draw[darr, red!60] (16.8, -1.5) -- (16.8, -3.3) -- (loss.east);

\node[anno, text=red!50!black] at (16.8, -2.2) {$\mathcal{L}_{\mathrm{task}}$};
\node[anno, text=red!50!black] at (9.5, -2.95) {$\lambda R(\hat{\mathbf{z}})$};

\end{tikzpicture}
}
\caption{The proposed \textit{mHC-Codec} framework.
The transmitter consists of an mHC-enhanced semantic encoder $f_\theta$ (detailed in Fig.~\ref{fig:mhc_dataflow}), a channel encoder $g_\phi$, and an entropy bottleneck (EB).
The EB quantizes the latent $\bar{\mathbf{z}}$ and provides a differentiable rate proxy $R(\hat{\mathbf{z}})$.
The channel supports both AWGN ($h{=}1$) and Rayleigh fading ($h{\sim}\mathcal{CN}(0,1)$).
At the receiver, a channel decoder $r_{\phi'}$ reconstructs the features, and a task head produces the output.
The two variants, \textit{mHC-Codec (Semantic)} and \textit{mHC-Codec (Task-Oriented)}, share the same transceiver backbone and differ only in the task head and loss.}
\label{fig:framework}
\end{figure*}
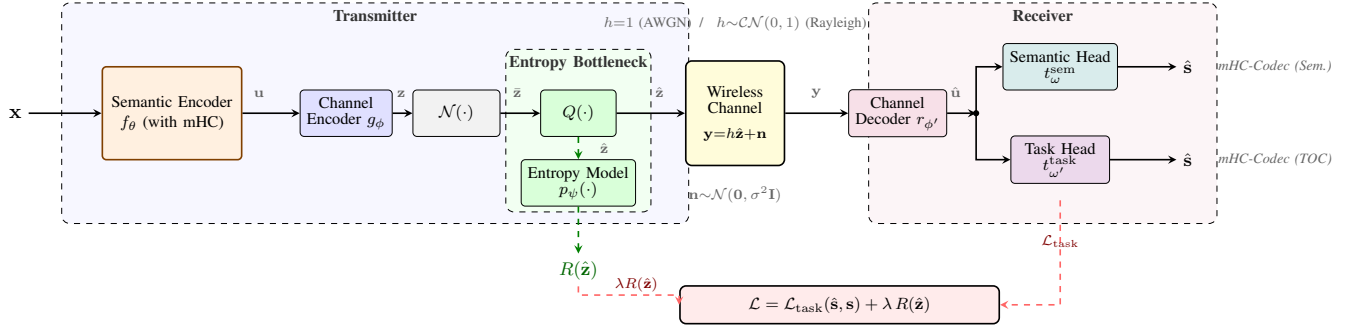


\subsection{System Models}
\label{sec:sys_overview}

As shown in Fig.~\ref{fig:framework}, we consider an end-to-end learned joint source-channel coding (JSCC) system with a transmitter, a wireless channel, and a receiver.
Given a source message $\mathbf{x}$ (e.g., a token sequence of length $N$) with task target $\mathbf{s}$, the system maps source symbols to channel symbols and operates as follows.

\subsubsection{Transmitter}
The transmitter contains a semantic encoder $f_{\theta}(\cdot)$ and a channel encoder $g_{\phi}(\cdot)$.
The semantic encoder extracts semantic/task-relevant features
\begin{align}
\mathbf{u} = f_{\theta}(\mathbf{x}),
\label{eq:sem_enc}
\end{align}
and the channel encoder maps them to real-valued channel codewords
\begin{align}
\mathbf{z} = g_{\phi}(\mathbf{u}).
\label{eq:ch_enc}
\end{align}
To meet a transmit power constraint, $\mathbf{z}$ is further normalized by $\mathcal{N}(\cdot)$, yielding $\bar{\mathbf{z}} = \mathcal{N}(\mathbf{z})$.

\subsubsection{Entropy bottleneck}
EB quantizes the normalized features and provides an estimated coding rate under a learned probability model~\cite{balle2018hyperprior}.
Let $Q(\cdot)$ denote scalar quantization and $\hat{\mathbf{z}} = Q(\bar{\mathbf{z}})$ the quantized features.
The entropy model $p_{\psi}(\cdot)$ estimates the probability mass function of $\hat{\mathbf{z}}$ for entropy coding.
The ideal coding length is
\begin{align}
R(\hat{\mathbf{z}}) = \sum_{i} -\log_2 p_{\psi}(\hat{z}_i),
\label{eq:rate_proxy}
\end{align}
where $i$ indexes all elements of $\hat{\mathbf{z}}$ across tokens and feature dimensions.

\subsubsection{Channel}
We consider a wireless channel with fading coefficient $h$.
The received signal is modeled as
\begin{align}
\mathbf{y} = h\,\hat{\mathbf{z}} + \mathbf{n}, \qquad
\mathbf{n}\sim\mathcal{N}(\mathbf{0},\sigma^2\mathbf{I}),
\label{eq:channel}
\end{align}
where noise $\mathbf{n}$ has a variance $\sigma^2$ evaluated by $\mathrm{SNR}_{\mathrm{dB}}$ under the adopted power normalization.
Each element of $\hat{\mathbf{z}}$ constitutes one real-valued \emph{channel use}; the total channel uses per source block is $k \times N$, where $k$ is the bottleneck dimension and $N$ is the sequence length. The corresponding \emph{bandwidth ratio} is $\rho = k/d$, where $d$ is the feature dimension. 
We consider three channel models:
\begin{itemize}
\item \emph{AWGN}: $h{=}1$ (no fading);
\item \emph{Rician fading}: $h{=}\sqrt{K/(K{+}1)} + h_{\mathrm{s}}$, where $h_{\mathrm{s}}{\sim}\mathcal{CN}(0,1/(K{+}1))$ is the scattered component, $K$ is the Rician $K$-factor, and $\mathbb{E}[|h|^2]{=}1$;
\item \emph{Rayleigh fading}: the special case $K{=}0$, i.e., $h{\sim}\mathcal{CN}(0,1)$.
\end{itemize}
We adopt a block-fading model in which $h$ is fixed within each transmission block (length $N$) and independently drawn across blocks.
Perfect CSI is assumed at the receiver, and the impact of imperfect receiver-side CSI will be studied in Sec.~\ref{sec:exp_channel_robust}.

\subsubsection{Receiver}
The receiver consists of a channel decoder $r_{\phi'}(\cdot)$ and a task head $t_{\omega}(\cdot)$.
The channel decoder reconstructs task-relevant features
\begin{align}
\hat{\mathbf{u}} = r_{\phi'}(\mathbf{y}),
\label{eq:ch_dec}
\end{align}
and the task head produces the final prediction
\begin{align}
\hat{\mathbf{s}} = t_{\omega}(\hat{\mathbf{u}}).
\label{eq:task_head}
\end{align}
Depending on the application, $t_{\omega}(\cdot)$ can be instantiated as a language modeling head (for SemCom) or a task module (for TOC).

\subsubsection{Rate-task optimization}
The system is trained end-to-end by minimizing a Lagrangian rate-distortion/task objective:
\begin{align}
\min_{\theta,\phi,\phi',\omega,\psi}\
\mathbb{E}\!\left[\mathcal{L}(\hat{\mathbf{s}},\mathbf{s})\right]
+ \lambda\,\mathbb{E}\!\left[R(\hat{\mathbf{z}})\right],
\label{eq:rt_objective}
\end{align}
where $\mathcal{L}(\hat{\mathbf{s}},\mathbf{s})$ quantifies the loss between reconstructed $\hat{\mathbf{s}}$ and real label $\mathbf{s}$, and $\lambda \ge 0$ controls the rate penalty.

\subsection{Preliminaries on Residual Connections, HC, and mHC}
\label{sec:prelim_hc_mhc}
Fig.~\ref{fig:three_arc} illustrates the three connection paradigms considered in this work: standard residual connections, HC, and mHC.

\paragraph{Standard Residual Connection}
Let $\mathcal{F}(\cdot;\mathbf{W}_l)$ denote the $l$-th Transformer block (e.g., attention + multi-layer perceptron (MLP) with normalization) parameterized by $\mathbf{W}_l$.
The standard residual connection~\cite{he2016deep} applies $\mathcal{F}$ to the input and adds the result to the input through an identity shortcut:
\begin{align}
\mathbf{x}_{l+1} = \mathbf{x}_l + \mathcal{F}(\mathbf{x}_l;\mathbf{W}_l),
\label{eq:residual}
\end{align}
where $\mathbf{x}_l \in \mathbb{R}^{d}$ is the representation at layer $l$.
Here, a stream refers to a feature pathway that carries hidden representations across layers.
The standard residual connection contains one such stream. 
This additive structure mitigates vanishing gradients and has become a standard building block in deep networks~\cite{he2016deep}.
However, all information is propagated through one residual pathway, which may limit representation diversity in deeply stacked networks~\cite{zhu2024hyperconnections}.

\paragraph{Hyper-Connections}
To address the single-stream limitation, HC~\cite{zhu2024hyperconnections} expands the hidden representation from one vector to $S$ parallel residual streams and introduces three learnable mixing components:
(i) a pre-mixing weight $\mathbf{H}^{\mathrm{pre}}_l$ that aggregates the streams before the sub-block,
(ii) a post-mixing weight $\mathbf{H}^{\mathrm{post}}_l$ that distributes the sub-block output back to the streams, and
(iii) a residual-path mixing matrix $\mathbf{H}^{\mathrm{res}}_l \in \mathbb{R}^{S \times S}$ that mixes the streams on the skip path.
The hidden representation at layer $l$ is $\mathbf{x}_l \in \mathbb{R}^{S \times d}$, and for a sequence of length $N$ it can be stacked as $\mathbf{X}_l \in \mathbb{R}^{N \times S \times d}$.
In the general HC formulation, all three components can be $S \times S$ matrices. In this work, $\mathbf{H}^{\mathrm{pre}}_l$ and $\mathbf{H}^{\mathrm{post}}_l$ are weight vectors in $\mathbb{R}^{S}$ that aggregate $S$ streams into a single view and distribute the output back, respectively.
The HC layer update is
\begin{align}
\mathbf{x}_{l+1}
&= \mathbf{H}^{\mathrm{res}}_l \mathbf{x}_l
+ \left(\mathbf{H}^{\mathrm{post}}_l\right)^{\!\top}\,
\mathcal{F}\!\left(\mathbf{H}^{\mathrm{pre}}_l \mathbf{x}_l;\mathbf{W}_l\right),
\label{eq:hc_update}
\end{align}
where $\mathbf{H}^{\mathrm{res}}_l\mathbf{x}_l$ mixes the $S$ streams on the skip path ($\mathbb{R}^{S \times d}\!\to\!\mathbb{R}^{S \times d}$), $\mathbf{H}^{\mathrm{pre}}_l\mathbf{x}_l$ aggregates the streams into a single input ($\mathbb{R}^{S \times d}\!\to\!\mathbb{R}^{d}$), and $(\mathbf{H}^{\mathrm{post}}_l)^{\!\top}$ distributes the sub-block output back to the $S$ streams ($\mathbb{R}^{d}\!\to\!\mathbb{R}^{S \times d}$).
\eqref{eq:hc_update} reduces to the standard residual form in~\eqref{eq:residual} when $S=1$ or when the mixing weights recover the single-stream case. Thus, HC can be viewed as a generalization of residual connections.

\paragraph{Manifold-Constrained Hyper-Connections}
While HC provides multiple information pathways across layers, unconstrained mixing can cause stream imbalance or collapse in deep networks~\cite{xie2025mhc}.
mHC addresses this issue by applying structured normalization to the mixing weights~\cite{xie2025mhc}.
The residual-path matrix $\mathbf{H}^{\mathrm{res}}_l \in \mathbb{R}^{S \times S}$ is constrained to DS manifold:
\begin{align}
\mathbf{H}\mathbf{1}=\mathbf{1},\quad
\mathbf{1}^{\top}\mathbf{H}=\mathbf{1}^{\top},\quad
H_{ij}\ge 0,\ \forall i,j,
\label{eq:ds_constraint}
\end{align}
where $\mathbf{1}$ is the all-ones vector.
The pre-mixing and post-mixing weights $\mathbf{H}^{\mathrm{pre}}_l$ and $\mathbf{H}^{\mathrm{post}}_l$ are softmax-normalized, so their entries are nonnegative and sum to one.
They therefore define stream-combination weights for input aggregation and output redistribution, respectively.
The DS constraint on $\mathbf{H}^{\mathrm{res}}_l$ makes the residual mixing mass-preserving: each row or each column sums to one. Thus, the streams are mixed through a convex redistribution across branches.
Together, DS-constrained residual mixing and softmax-normalized pre/post weights improve robustness and reduce stream collapse in deep networks.
The theoretical properties of DS mixing related to rate efficiency will be analyzed in Sec.~\ref{sec:theory}.


\section{Proposed Method: \textit{mHC coding scheme}}
\label{sec:proposed}

\begin{figure}[t]
\centering
\resizebox{0.7\columnwidth}{!}{%
\begin{tikzpicture}[
    >=Stealth,
    font=\small,
    arr/.style={->, thick, >=Stealth},
    procbox/.style={draw, thick, rounded corners=3pt, minimum height=0.7cm, align=center},
    fblock/.style={procbox, fill=cyan!10, minimum width=2.2cm, minimum height=0.85cm},
    dsblock/.style={procbox, fill=red!10, draw=red!50!black, minimum width=2.0cm, minimum height=0.8cm},
    smblock/.style={procbox, fill=blue!10, draw=blue!40!black, minimum width=2.0cm, minimum height=0.8cm},
    expandbox/.style={procbox, fill=green!8, draw=green!50!black, minimum width=2.6cm, minimum height=0.75cm},
    normbox/.style={procbox, fill=gray!8, minimum width=2.2cm},
    sumnode/.style={draw, thick, circle, minimum size=0.55cm, inner sep=0pt, fill=yellow!15, font=\normalsize},
    dimtag/.style={font=\scriptsize, text=teal!70!black, fill=teal!6, draw=teal!25,
                   rounded corners=1.5pt, inner xsep=3pt, inner ysep=1.5pt},
    anno/.style={font=\scriptsize, text=gray!55!black},
]

\def\xleft{-2.8}
\def\xright{2.8}
\def\xmid{0}

\node[font=\normalsize] (input) at (\xmid, 0) {$\mathbf{x} \in \mathbb{R}^{d}$};

\node[expandbox] (expand) at (\xmid, -1.1)
    {Stream Expansion\\[-2pt]{\scriptsize replicate $\times S$}};
\draw[arr] (input) -- (expand);

\node[dimtag] (expanddim) at (\xmid, -1.95) {$\mathbb{R}^{d} \to \mathbb{R}^{S \times d}$};
\draw[arr] (expand) -- (expanddim);

\def\ylin{-2.8}
\node[font=\normalsize] (xl) at (\xmid, \ylin) {$\mathbf{x}_l \in \mathbb{R}^{S \times d}$};
\draw[arr] (expanddim) -- (xl);

\def\yfork{-3.4}
\coordinate (fork) at (\xmid, \yfork);
\draw[thick] (xl) -- (fork);
\draw[arr] (fork) -| (\xleft, -3.9);
\draw[arr] (fork) -| (\xright, -3.9);

\def\yhres{-4.5}
\node[dsblock] (Hres) at (\xleft, \yhres)
    {$\mathbf{H}^{\mathrm{res}}_l$\\[-1pt]
     {\scriptsize DS,\; $S{\times}S$}};

\node[dimtag] (resdim) at (\xleft, -5.5) {$\mathbb{R}^{S \times d}$};
\draw[arr] (Hres) -- (resdim);

\def\ysum{-10.2}
\draw[arr, red!40!black] (\xleft, -5.75) -- (\xleft, \ysum);

\node[anno, rotate=90, anchor=south, text=red!45!black]
    at ({\xleft-0.45}, -7.8) {\small residual path};


\def\yhpre{-4.5}
\node[smblock] (Hpre) at (\xright, \yhpre)
    {$\mathbf{H}^{\mathrm{pre}}_l$\\[-1pt]
     {\scriptsize SM,\; $\mathbb{R}^{S}$}};

\node[dimtag] (predim) at (\xright, -5.4) {$\mathbb{R}^{S \times d} \!\to\! \mathbb{R}^{d}$};
\draw[arr] (Hpre) -- (predim);

\def\yF{-6.4}
\node[fblock] (Fblock) at (\xright, \yF)
    {$\mathcal{F}(\cdot;\,\mathbf{W}_l)$\\[-1pt]
     {\scriptsize Attn / MLP}};
\draw[arr] (\xright, -5.6) -- (Fblock);

\node[dimtag] (fdim) at (\xright, -7.35) {$\mathbb{R}^{d}$};
\draw[arr] (Fblock) -- (fdim);

\def\yhpost{-8.2}
\node[smblock] (Hpost) at (\xright, \yhpost)
    {$\mathbf{H}^{\mathrm{post}}_l$\\[-1pt]
     {\scriptsize SM,\; $\mathbb{R}^{S}$}};
\draw[arr] (fdim) -- (Hpost);

\node[dimtag] (postdim) at (\xright, -9.15) {$\mathbb{R}^{d} \!\to\! \mathbb{R}^{S \times d}$};
\draw[arr] (Hpost) -- (postdim);

\node[anno, rotate=90, anchor=south, text=blue!45!black]
    at (4.4, -7.8) {\small branch path};

\node[sumnode] (sum) at (\xmid, \ysum) {$+$};
\draw[arr, red!40!black] (\xleft, \ysum) -- (sum);
\draw[arr, blue!40!black] (\xright, -9.4) -- (\xright, \ysum) -- (sum);

\def\yxout{-11.0}
\node[font=\normalsize] (xlp1) at (\xmid, \yxout) {$\mathbf{x}_{l+1} \in \mathbb{R}^{S \times d}$};
\draw[arr] (sum) -- (xlp1);

\begin{pgfonlayer}{background}
\node[draw, thick, rounded corners=7pt, densely dashed, gray!50,
      fill=gray!2, inner xsep=12pt, inner ysep=10pt,
      fit=(xl)(Hres)(Hpre)(Fblock)(Hpost)(sum)(xlp1)
      (resdim)(postdim)(predim)(fdim)] (layerbox) {};
\end{pgfonlayer}

\node[font=\small\bfseries, text=gray!65!black, anchor=north west]
    at ([yshift=2pt]layerbox.north west)
    {\;\;mHC Layer $l$\;\;($\times\, L$)};

\def\yrepeat{-11.8}
\node[font=\large, text=gray] at (\xmid, \yrepeat) {$\vdots$};

\def\yln{-12.5}
\node[normbox] (ln) at (\xmid, \yln) {LayerNorm};
\draw[arr] (\xmid, -12.1) -- (ln);

\def\yreduce{-13.6}
\node[expandbox, fill=orange!8, draw=orange!50!black] (reduce) at (\xmid, \yreduce)
    {Stream Reduction\\[-2pt]{\scriptsize $\mathbf{y} = \textstyle\sum_{s}\mathbf{x}_{L,s}$}};
\draw[arr] (ln) -- (reduce);

\node[dimtag] (reducedim) at (\xmid, -14.45) {$\mathbb{R}^{S \times d} \to \mathbb{R}^{d}$};
\draw[arr] (reduce) -- (reducedim);

\def\yout{-15.2}
\node[font=\normalsize] (output) at (\xmid, \yout)
    {$\mathbf{y} \in \mathbb{R}^{d}$};
\draw[arr] (reducedim) -- (output);

\end{tikzpicture}
}
\caption{Dataflow of the mHC-enhanced semantic encoder. Input token embeddings are replicated into $S$ streams ($\mathbb{R}^d \to \mathbb{R}^{S \times d}$). At each layer, $\mathbf{H}^{\mathrm{res}}_l$ (DS-constrained, $S{\times}S$) mixes the residual streams, while $\mathbf{H}^{\mathrm{pre}}_l$ (softmax) aggregates them into a single input for the Transformer sub-block $\mathcal{F}$, and $\mathbf{H}^{\mathrm{post}}_l$ (softmax) distributes the output back to $S$ streams. After $L$ layers, the streams are reduced to $\mathbb{R}^d$ via summation across streams.}
\label{fig:mhc_dataflow}
\end{figure}
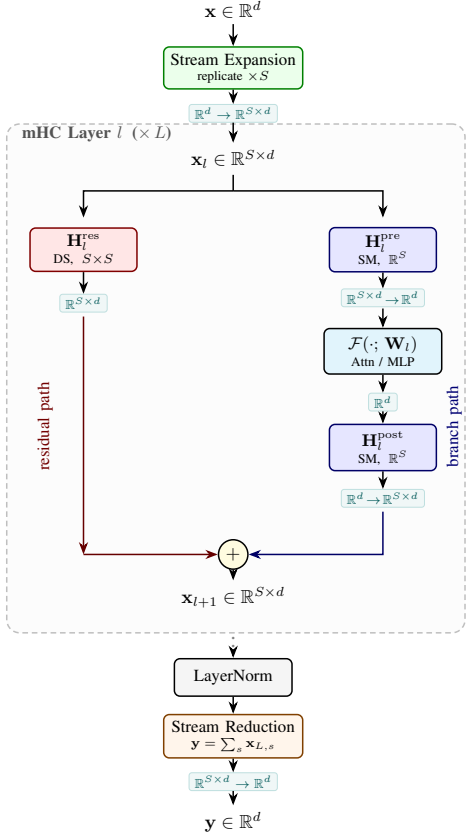

In what follows, we will give details on the proposed \textit{mHC coding scheme}, including the mHC-enhanced semantic encoder, EB, task instantiations, training procedure, and resource efficiency analysis.

\subsection{mHC-Enhanced Semantic Encoder}
\label{sec:method_encoder}

Fig.~\ref{fig:mhc_dataflow} illustrates the mHC-enhanced semantic encoder $f_{\theta}(\cdot)$, which consists of stream expansion, $L$ mHC Transformer layers, and stream reduction.

\subsubsection{Stream expansion}
The input token embeddings $\mathbf{x}\in\mathbb{R}^{d}$ are replicated across $S$ streams to form $\mathbf{x}_0\in\mathbb{R}^{S \times d}$.
All mixing matrices are initialized to the identity (or uniform weights for softmax-normalized ones). Thus, the model starts from the standard residual form in \eqref{eq:residual} and smoothly transitions into multi-stream mixing during training. 

\subsubsection{mHC layer update}
As shown in Fig.~\ref{fig:mhc_dataflow}, each layer $l$ splits $\mathbf{x}_l\in\mathbb{R}^{S \times d}$ into two parallel paths.
The \emph{residual path} applies the DS-constrained matrix $\mathbf{H}^{\mathrm{res}}_l\in\mathbb{R}^{S \times S}$, which mixes the $S$ streams while preserving dimensionality ($\mathbb{R}^{S \times d}\!\to\!\mathbb{R}^{S \times d}$).
The \emph{branch path} first aggregates the $S$ streams into a single representation via the softmax-normalized weight vector $\mathbf{H}^{\mathrm{pre}}_l\in\mathbb{R}^{S}$ ($\mathbb{R}^{S \times d}\!\to\!\mathbb{R}^{d}$), passes it through the Transformer sub-block $\mathcal{F}(\cdot;\mathbf{W}_l)$ ($\mathbb{R}^{d}\!\to\!\mathbb{R}^{d}$), and distributes the output back to $S$ streams via $\mathbf{H}^{\mathrm{post}}_l\in\mathbb{R}^{S}$ ($\mathbb{R}^{d}\!\to\!\mathbb{R}^{S \times d}$).
The two paths are summed element-wise to produce the next-layer representation using \eqref{eq:hc_update}.
In practice, $\mathbf{H}^{\mathrm{res}}_l$ is parameterized through unconstrained logits $\mathbf{A}^{\mathrm{res}}_l$ and projected onto the DS manifold via Sinkhorn-Knopp algorithm~\cite{sinkhorn1967concerning}:
$\mathbf{H}^{\mathrm{res}}_l = \Pi_{\mathcal{DS}}(\mathbf{A}^{\mathrm{res}}_l)$ as in \eqref{eq:ds_constraint},
while $\mathbf{H}^{\mathrm{pre}}_l = \mathrm{softmax}(\mathbf{A}^{\mathrm{pre}}_l)$ and
$\mathbf{H}^{\mathrm{post}}_l = \mathrm{softmax}(\mathbf{A}^{\mathrm{post}}_l)$.
All projections are differentiable and add negligible overhead since $S \ll d$.

\subsubsection{Stream reduction}
After the final encoder layer and LayerNorm, the $S$ streams are aggregated into a single representation for the channel encoder $g_{\phi}(\cdot)$ via summation across streams:
$\mathbf{y} = \sum_{s=1}^{S}\mathbf{x}_{L,s}$,
where $\mathbf{x}_{L,s}\in\mathbb{R}^{d}$ is the $s$-th stream of the final-layer output.
This ensures that the downstream channel encoder and EB operate on a single-stream features of dimension $d$, preserving full compatibility with standard entropy-coding pipelines and introducing no additional transmission symbols.

\subsection{Rate-Aware Entropy Bottleneck}
\label{sec:method_eb}

The EB, introduced as part of the system model in Sec.~\ref{sec:sys_overview}, plays a pivotal role in \textit{mHC coding scheme} by linking the encoder design to explicit rate control.
Thus we will discuss the EB design in our framework.

\subsubsection{Quantization and straight-through estimation}
During training, quantization $Q(\cdot)$ is non-differentiable.
Following~\cite{balle2018hyperprior}, we replace nearest-integer rounding with additive uniform noise $\hat{\mathbf{z}} = \bar{\mathbf{z}} + \mathbf{e}$, $\mathbf{e}\sim\mathcal{U}(-\tfrac{1}{2},\tfrac{1}{2})$, which provides a differentiable approximation to scalar quantization and allows gradients to be propagated through the rate $R(\hat{\mathbf{z}})$ in~\eqref{eq:rate_proxy}.
At inference, true rounding is applied.

\subsubsection{Learned entropy model}
The parametric entropy model $p_{\psi}(\cdot)$ is implemented as a fully factorized density whose parameters are jointly optimized with the rest of the network through the rate loss $\lambda\,\mathbb{E}[R(\hat{\mathbf{z}})]$.
Because the rate term directly penalizes the coding cost, the encoder is encouraged to produce features, of which marginal statistics are well-matched by $p_{\psi}$, yielding compact entropy-coded representations.

\textit{Remark:}
We define $k$ as the bottleneck dimension. The bottleneck dimension $k$ fixes the number of channel uses for wireless transmission, while EB controls the entropy-coded bit cost of the quantized channel features. Thus, EB improves transmission efficiency by lowering the entropy-coded bit cost of the quantized features, rather than by reducing the dimensionality.

\subsection{Task Instantiations}
\label{sec:method_instantiations}

\subsubsection{\textit{mHC-SemCom}}
The receiver uses a semantic head $t^{\mathrm{sem}}_{\omega}(\cdot)$ (e.g., a linear projection followed by softmax over the vocabulary) to produce token-level distributions.
The task loss is the cross-entropy between the predicted and ground-truth token distributions:
\begin{align}
\mathcal{L}_{\mathrm{sem}} = -\frac{1}{N}\sum_{i=1}^{N}\log p_{\omega}(s_i \mid \hat{\mathbf{u}}),
\label{eq:loss_sem}
\end{align}
where $N$ is the sequence length and $s_i$ is the $i$-th ground-truth token.

\subsubsection{\textit{mHC-TOC}}
The receiver uses a task head $t^{\mathrm{task}}_{\omega}(\cdot)$ tailored to a downstream objective (e.g., classification).
The task loss is a standard supervised loss:
\begin{align}
\mathcal{L}_{\mathrm{TOC}} = -\sum_{c=1}^{C} y_c \log \hat{y}_c,
\label{eq:loss_task}
\end{align}
where $\hat{y}_c = t^{\mathrm{task}}_{\omega}(\hat{\mathbf{u}})_c$ is the predicted probability for class $c$ and $y_c$ is the one-hot label.

\subsection{SemCom Pre-training and TOC Adaptation}
\label{sec:method_train}

The proposed framework is trained and evaluated through a two-stage procedure, as summarized in Algorithm~\ref{alg:training}.
Stage~1 performs end-to-end SemCom pre-training.
Specifically, the mHC semantic encoder, channel encoder, EB, channel decoder, and semantic head are jointly optimized by minimizing the rate-semantic objective $\mathcal{L}_{\mathrm{sem}} + \lambda\,R(\hat{\mathbf{z}})$ in~\eqref{eq:rt_objective}.
This stage learns rate-aware semantic representations that are robust to wireless channel impairments.
Stage~2 adapts the learned representation to TOC.
The pretrained SemCom transceiver is frozen, and only a lightweight task head $t^{\mathrm{task}}_{\omega'}(\cdot)$ is trained with $\mathcal{L}_{\mathrm{TOC}}$ in~\eqref{eq:loss_task}.
Thus, SemCom uses the end-to-end training in Stage~1, whereas TOC adds a task-head training stage on top of the frozen SemCom representation.

\begin{algorithm}[t]
\caption{SemCom Pre-training and TOC Adaptation of the \textit{mHC coding scheme}}
\label{alg:training}
\begin{algorithmic}[1]
\Statex \textit{Stage~1: SemCom Pre-training} (update all parameters)
\Require Input $\mathbf{x}$, target $\mathbf{s}$, channel $h$, $\mathrm{SNR}_{\mathrm{dB}}$, weight $\lambda$
\State $\mathbf{u} \gets f_{\theta}(\mathbf{x})$
    \Comment{mHC semantic encoder~\eqref{eq:hc_update}}
\State $\mathbf{z} \gets g_{\phi}(\mathbf{u})$
    \Comment{Channel encoder}
\State $\bar{\mathbf{z}} \gets \mathcal{N}(\mathbf{z})$
    \Comment{Power normalization}
\State $\hat{\mathbf{z}},\, R(\hat{\mathbf{z}}) \gets \mathrm{EB}(\bar{\mathbf{z}};\psi)$
    \Comment{Quantize and rate estimate}
\State $\mathbf{y} \gets h\,\hat{\mathbf{z}} + \mathbf{n}$,\;
    $\mathbf{n}\!\sim\!\mathcal{N}(\mathbf{0},\sigma^2\mathbf{I})$
    \Comment{Channel~\eqref{eq:channel}}
\State $\hat{\mathbf{u}} \gets r_{\phi'}(\mathbf{y})$
    \Comment{Channel decoder}
\State $\hat{\mathbf{s}} \gets t^{\mathrm{sem}}_{\omega}(\hat{\mathbf{u}})$
    \Comment{Semantic head}
\State $\mathcal{L} \gets \mathcal{L}_{\mathrm{sem}} + \lambda\, R(\hat{\mathbf{z}})$
    \Comment{Eqs.~\eqref{eq:loss_sem},\,\eqref{eq:rt_objective}}
\State Update $\{\theta, \phi, \psi, \phi', \omega\}$ via $\nabla\mathcal{L}$
\Statex
\Statex \textit{Stage~2: TOC Head Training} (freeze pretrained SemCom transceiver)
\State Execute Steps~1\,--\,6 with frozen parameters
\State $\hat{\mathbf{s}} \gets t^{\mathrm{task}}_{\omega'}(\hat{\mathbf{u}})$
    \Comment{Task head}
\State $\mathcal{L} \gets \mathcal{L}_{\mathrm{TOC}}$
    \Comment{Eq.~\eqref{eq:loss_task}}
\State Update $\omega'$ via $\nabla\mathcal{L}$
\end{algorithmic}
\end{algorithm}

\subsection{Resource Efficiency Analysis}
\label{sec:method_resource}

Since resource efficiency is a central design goal of \textit{mHC coding scheme}, we systematically analyze the computational, communication, and memory overhead introduced by mHC relative to the standard residual connection baseline.

\subsubsection{Computational overhead (FLOPs)}
The dominant computation in each Transformer block is self-attention and the feed-forward network, both scaling as $\mathcal{O}(Nd^2)$ per layer, where $N$ is the sequence length and $d$ is the feature dimension.
mHC introduces three stream-mixing operations per layer, corresponding to pre-mixing, post-mixing, and residual-path mixing.
Their total cost scales as $\mathcal{O}(NSd)$, since the mixing is applied over the $S$ streams for each token and feature dimension.
Since $S \ll d$ (e.g., $S{=}4$, $d{=}288$), the relative overhead is typically ${<}\,3\%$:
\begin{equation}
\frac{\Delta\mathrm{FLOPs}}{\mathrm{FLOPs}_{\mathrm{baseline}}}
= \frac{\mathcal{O}(NSd)}{\mathcal{O}(Nd^2)}
= \mathcal{O}\!\left(\frac{S}{d}\right).
\label{eq:flops_ratio}
\end{equation}
Sinkhorn-Knopp projection operates on $S\times S$ matrices and has a computational complexity of $\mathcal{O}(S^2 t_{\max})$ per layer (e.g., $t_{\max}{=}20$), which is independent of $d$ and $N$ and is negligible compared with the $\mathcal{O}(Nd^2)$ cost of self-attention and feed-forward computation in each Transformer block.

\subsubsection{Parameter overhead}
Each layer introduces three $S\times S$ learnable parameters for the three unconstrained mixing matrices (before DS projection).
For an $L$-layer encoder, the total parameter overhead is $3S^2 L$.
Compared to the backbone parameter count of $\mathcal{O}(Ld^2)$, the relative increase is
\begin{equation}
\frac{\Delta\mathrm{Params}}{\mathrm{Params}_{\mathrm{baseline}}}
= \frac{\mathcal{O}(S^2 L)}{\mathcal{O}(Ld^2)}
= \mathcal{O}\!\left(\frac{S^2}{d^2}\right),
\label{eq:param_ratio}
\end{equation}
which is vanishingly small (e.g., ${\approx}\,0.1\%$ for $S=4$, $d=128$).

\subsubsection{Transmission rate overhead}
The multi-stream representation exists only inside the semantic encoder.
After the final encoder layer, the $S$ streams are combined into a single stream (cf.\ Sec.~\ref{sec:method_encoder}) before entering the channel encoder and EB.
Therefore, the channel features $\hat{\mathbf{z}}$ have the same dimensionality as in the baseline, and no extra symbols are transmitted.
Moreover, as will be shown in Sec.~\ref{sec:theory}, DS mixing is entropy non-increasing. Thus, the ideal coding length does not increase due to the multi-stream architecture.

\subsubsection{Memory footprint}
The multi-stream representation $\mathbf{x}_l\in\mathbb{R}^{S\times d}$ is maintained only inside the $L$-layer encoder; all downstream components (channel encoder, EB, channel, decoder) operate on the reduced single-stream features $\mathbf{y}\in\mathbb{R}^{d}$.
Therefore, the encoder activation memory scales as $\mathcal{O}(S)$ relative to the baseline, while the rest of the pipeline remains unchanged~\cite{zhu2024hyperconnections}.

\subsubsection{Summary}
Table~\ref{tab:resource_summary} summarizes the resource comparison.
The mHC architecture provides richer representational capacity at the minimal cost across all resource dimensions.

\begin{table}[t]
\centering
\caption{Resource overhead of mHC relative to a residual-connection baseline.}
\label{tab:resource_summary}
\begin{tabular}{l c c}
\toprule
Resource Dimension & Overhead & Depends on \\
\midrule
FLOPs per layer & $\mathcal{O}(S/d)$ & $S, d$ \\
Parameters per layer & $\mathcal{O}(S^2/d^2)$ & $S, d$ \\
Transmitted symbols & $0$ (unchanged) & -- \\
Ideal coding cost & $\le 0$ (Corollary~\ref{cor:entropy_noninc_main}) & DS property \\
Encoder activation memory & ${\sim}\, S\times$ & $S$ \\
\bottomrule
\end{tabular}
\end{table}

\section{Theoretical Analysis}
\label{sec:theory}

This section establishes the rate-theoretic properties of DS-constrained stream mixing.
We first prove that DS mixing is entropy non-increasing (Sec.~\ref{sec:theory_entropy}), then connect this result to practical EB rate behavior (Sec.~\ref{sec:theory_rate}), and finally analyze its implication for the rate--distortion tradeoff (Sec.~\ref{sec:theory_rd}).

\subsection{Entropy Non-Increasing Property of DS Mixing}
\label{sec:theory_entropy}

Consider the multi-stream representation $\bZ\in\RR^{S\times d}$ at a given encoder layer.
The DS-constrained residual mixing produces $\widetilde{\bZ}=\bH\bZ$, where $\bH\in\DS_S$ acts on the stream dimension (i.e., each column of $\bZ$ undergoes the same linear mixing).
We impose the following mild regularity condition.

\begin{assumption}[Non-singularity]
\label{assump:invertible_main}
The mixing matrix $\bH\in\DS_S$ is non-singular, i.e., $\det(\bH)\neq 0$.
\end{assumption}

Assumption~\ref{assump:invertible_main} is mild: singular DS matrices form a measure-zero subset of $\DS_S$. Thus a randomly initialized or learned $\bH$ is non-singular with probability one.

\begin{theorem}[Differential entropy under linear stream mixing]
\label{thm:lin_entropy_main}
Let $\bZ\in\RR^{S\times d}$ be a continuous random matrix with well-defined differential entropy, and let $\widetilde{\bZ}=\bH\bZ$ with $\bH\in\RR^{S\times S}$ satisfying Assumption~\ref{assump:invertible_main}.
Then
\begin{equation}
h\!\big(\vecop(\widetilde{\bZ})\big)
= h\!\big(\vecop(\bZ)\big) + d\log|\det(\bH)|,
\label{eq:lin_entropy_main}
\end{equation}
where $\vecop(\cdot)$ denotes columnwise vectorization and $h(\cdot)$ denotes differential entropy.
\end{theorem}

\begin{IEEEproof}
See Appendix~\ref{app:proof_lin_entropy}.
\end{IEEEproof}

\begin{theorem}[Determinant bound for DS matrices]
\label{thm:det_bound_main}
For any $\bH\in\DS_S$, $|\det(\bH)|\le 1$.
Equality holds if and only if $\bH$ is a permutation matrix.
\end{theorem}

\begin{IEEEproof}
See Appendix~\ref{app:proof_det_bound}.
\end{IEEEproof}

Combining Theorems~\ref{thm:lin_entropy_main} and~\ref{thm:det_bound_main} yields the following.

\begin{corollary}[Entropy non-increasing under DS mixing]
\label{cor:entropy_noninc_main}
Under Assumption~\ref{assump:invertible_main}, if $\bH\in\DS_S$, then
\begin{equation}
h\!\big(\vecop(\widetilde{\bZ})\big)
\le h\!\big(\vecop(\bZ)\big),
\end{equation}
with equality if and only if $\bH$ is a permutation matrix.
\end{corollary}

\subsection{Implications for Rate Behavior}
\label{sec:theory_rate}

We now connect the entropy non-increasing property to the rate behavior of the EB.
Let $\hat{\bm z}$ denote the quantized channel features transmitted over the wireless link.

\subsubsection{Ideal vs.\ model-based rate}
Under ideal entropy coding with the true mass function $p(\hat{\bm z})$, the expected code length equals the entropy:
\begin{equation}
R_{\mathrm{ideal}} = H(\hat{\bm z}).
\end{equation}
In practice, the EB employs a learned parametric model $q(\hat{\bm z})$, and the realized rate satisfies
\begin{equation}
R_{\mathrm{model}}
= H(\hat{\bm z}) + D_{\mathrm{KL}}(p\|q),
\label{eq:model_rate_main}
\end{equation}
where $D_{\mathrm{KL}}(p\|q)\ge 0$ quantifies the mismatch between the true distribution $p$ and modeled distribution $q$.

\subsubsection{Rate preservation under DS mixing}
Corollary~\ref{cor:entropy_noninc_main} shows that DS mixing does not increase the differential entropy of the continuous features.
Under the standard high-resolution quantization approximation~\cite{gray1998quantization}, the discrete entropy $H(\hat{\bm z})$ is well-approximated by the differential entropy plus a quantization-dependent constant.
This implies that DS-constrained mixing does not incur an ideal-rate overhead.

\subsubsection{Stream-sum preservation}
The DS property $\bone^\top \bH=\bone^\top$ ensures that the column sums of $\bZ$ are invariant under mixing:
\begin{equation}
\bone^\top \widetilde{\bZ} = \bone^\top \bH \bZ = \bone^\top \bZ.
\label{eq:sum_invariance_main}
\end{equation}
This is consistent with the stream reduction $\mathbf{y}=\sum_{s=1}^{S}\mathbf{x}_{L,s}$: the aggregated representation seen by the channel encoder is preserved regardless of how the DS mixing redistributes information across streams.

\subsubsection{Practical rate reduction}
Beyond the ideal-rate guarantee, DS mixing can also reduce the model-based rate $R_{\mathrm{model}}$ in~\eqref{eq:model_rate_main}.
Since each row of $\bH\in\DS_S$ defines a convex combination, the mixed streams tend to exhibit more balanced statistics across dimensions.
This improved regularity can reduce the mismatch $D_{\mathrm{KL}}(p\|q)$ when the EB uses a factorized or weakly-coupled prior, lowering $R_{\mathrm{model}}$ even when $H(\hat{\bm z})$ changes only mildly.

\subsubsection{Resource efficiency summary}
Combining the above rate results with the computational analysis in Sec.~\ref{sec:method_resource}:
(i)~the additional computation and parameters introduced by DS mixing scale as $\mathcal{O}(S/d)$ and $\mathcal{O}(S^2/d^2)$ relative to the standard residual-connection backbone, respectively;
(ii)~stream reduction eliminates multi-stream dimensionality before transmission;
and (iii)~DS mixing does not increase the ideal coding cost (Corollary~\ref{cor:entropy_noninc_main}).
The richer representational capacity provided by multi-stream mixing thus comes at negligible cost in both computation and communication rate.

\subsection{Rate-Distortion Dominance}
\label{sec:theory_rd}

We analyze the achievable rate-distortion (RD) tradeoff by comparing the optimal distortion attainable by the baseline and mHC model classes under the same constraint on the expected EB rate.
Let $D$ denote the expected task distortion and $R$ the expected rate induced by the EB.
The achievable RD function of a model family $\mathcal{F}$ is defined as
\begin{equation}
D^\star_{\mathcal{F}}(R)
\triangleq
\inf_{\substack{f\in\mathcal{F}\\ \mathbb{E}[\mathrm{Rate}]\le R}}
\mathbb{E}\!\left[\ell(\widehat{Y},Y)\right].
\end{equation}

\begin{lemma}[Baseline]
\label{lem:baseline_inclusion_main}
Let $\mathcal{F}_{\mathrm{base}}$ and $\mathcal{F}_{\mathrm{mHC}}$ denote the hypothesis classes of the baseline (standard residual connection) and mHC architectures, respectively.
Since $\bI\in\DS_S$, choosing $\bH^{\mathrm{res}}_l=\bI$ for all layers and assigning the pre/post-mixing weights to one stream makes the mHC layer reduce to the standard residual layer, up to a fixed scaling that can be absorbed by the channel encoder parameters.
Hence $\mathcal{F}_{\mathrm{base}}\subseteq\mathcal{F}_{\mathrm{mHC}}$.
\end{lemma}

\begin{theorem}[RD dominance]
\label{thm:rd_dominance_main}
Under Lemma~\ref{lem:baseline_inclusion_main},
\begin{equation}
D^\star_{\mathrm{mHC}}(R) \le D^\star_{\mathrm{baseline}}(R),\qquad \forall\, R\ge 0.
\end{equation}
Equivalently, for any target distortion level $D$,
\begin{equation}
R^\star_{\mathrm{mHC}}(D) \le R^\star_{\mathrm{baseline}}(D).
\end{equation}
\end{theorem}

\begin{IEEEproof}
See Appendix~\ref{app:proof_rd_dominance}.
\end{IEEEproof}

\textit{Remark:}
Theorem~\ref{thm:rd_dominance_main} is an idealized result that assumes globally optimal solutions over the entire hypothesis class.
In practice, the observed improvements arise from the combined effects of improved optimization stability, more balanced feature statistics for the factorized entropy model, and the non-expansive nature of DS mixing.

\section{Numerical Results}
\label{sec:experiments}

\begin{figure*}[t]
\centering
\includegraphics[width=0.8\textwidth]{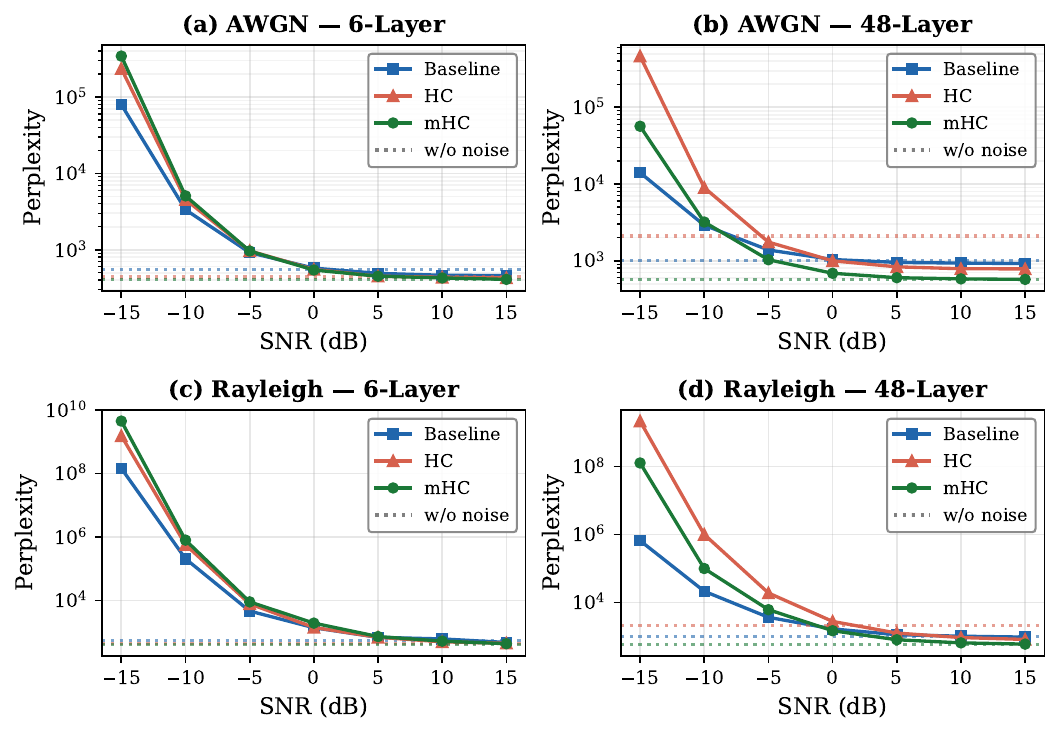}
\caption{FineWeb10B perplexity vs.\ SNR under AWGN and Rayleigh fading for 6-layer and 48-layer models ($y$-axis log-scaled).
All three methods use the entropy bottleneck (EB) with $\lambda{=}0.01$.
mHC achieves the lowest PPL at moderate-to-high SNR ($\ge 0$\,dB), with the advantage most pronounced in the 48-layer setting.}
\label{fig:snr_ppl}
\end{figure*}

\subsection{Experimental Setups}

\subsubsection{Tasks and datasets}
We evaluate mHC coding scheme on both SemCom and TOC settings.
For SemCom, all models are trained on FineWeb10B~\cite{penedo2024fineweb}, a large-scale web text corpus pre-tokenized with  GPT-2 tokenizer (vocabulary size $50{,}304$), and additionally evaluated on WikiText-103~\cite{merity2017wikitext103} for out-of-distribution generalization.
For TOC, the semantic encoder is frozen after SemCom pre-training, and a linear classification head is trained on
AG News~\cite{zhang2015agnews} (4-class topic classification).

\subsubsection{Baselines}
We compare the proposed mHC against four baseline schemes organized into two groups.

\textit{Group~1: Connection-scheme comparison.}
All methods share the same Transformer backbone, EB, channel model, and training procedure; they differ \emph{only} in the inter-layer connectivity: (i)~\emph{Baseline~(EB)}: standard residual connections ($S=1$), (ii)~\emph{HC~(EB)}: standard HC~\cite{zhu2024hyperconnections} with $S=4$ streams and unconstrained mixing matrices, and (iii)~the proposed \emph{mHC~(EB)} with $S=4$ streams and DS-constrained mixing (Sinkhorn-Knopp projection, $10$ iterations, $\tau{=}0.05$).

\textit{Group~2: Channel-coding component comparison.}
To isolate the contribution of the EB-based rate control, we introduce two additional baselines that replace the EB with alternative channel-coding strategies while keeping the same Transformer backbone ($S{=}1$, standard residual connections) and identical training budget:
(iv)~\emph{Dense~(DeepSC)}: a DeepSC-style dense channel coding scheme~\cite{xie2021deepsc} where the channel encoder directly maps semantic features to channel symbols without quantization or rate control; power normalization ensures unit average transmit power;
and (v)~\emph{VIB}: a VIB-based channel coding scheme~\cite{shao2022ib_edge_inference} where the channel encoder outputs mean and log-variance parameters, and the rate is measured as KL divergence between the features posterior and a standard Gaussian prior.
Both baselines use the same bottleneck dimension $k{=}64$ (and thus the same bandwidth ratio $\rho$) and are trained under AWGN uniformly from $\mathcal{U}[5,15]$~dB.
This component-level comparison ensures a fair evaluation. All methods share the same backbone architecture, parameter budget, and the number of channel uses, differing \emph{only} in the channel-coding module (EB vs.\ dense vs.\ VIB).

\subsubsection{Performance metrics}
We report cross-entropy (CE) loss and perplexity ($\mathrm{PPL}=\exp(\mathrm{CE})$) for language modeling, classification accuracy for TOC, and entropy-coded bit rate $R(\hat{\mathbf{z}})=-\sum_i \log_2 p_\psi(\hat{z}_i)$ (bits per sample) for rate evaluation.
The bandwidth ratio is defined as $\rho = k/d$, where $k$ is the channel bottleneck dimension and $d$ is the Transformer hidden size.
We fix $k=64$ for all EB-based methods. Thus, $\rho$ is the same across methods within each depth setting, and equals $64/288\approx 0.22$ for 6-layer models and $64/150\approx 0.43$ for 48-layer models.
All metrics are evaluated across SNR $\in\{-15,-10,-5,0,5,10,15\}$~dB under AWGN, Rayleigh fading, and Rician fading channels, as well as a clean (noise-free) condition.
Additionally, robustness under imperfect CSI is evaluated in Sec.~\ref{sec:exp_channel_robust}.

\subsubsection{Implementation details}
We adopt a GPT-2--style causal Transformer~\cite{radford2019gpt2} following nanoGPT~\cite{karpathy2023nanogpt}, with two scales at ${\approx}\,20$\,M parameters:
a 6-layer model ($d{=}288$, $6$ heads) and a 48-layer model ($d{=}150$, $6$ heads), both with context length $1{,}024$.
The channel encoder compresses the semantic features to dimension $k{=}64$ via a linear projection; the EB uses a factorized Gaussian prior.
Training uses AWGN with SNR sampled uniformly from $\mathcal{U}[5,15]$~dB at each iteration to promote robustness across a range of link qualities; $\lambda{=}0.01$, AdamW~\cite{loshchilov2017decoupled} ($\beta_1{=}0.9$, $\beta_2{=}0.95$, weight decay $0.1$), cosine LR schedule (peak $6{\times}10^{-4}$, $200$-step warm-up), bfloat16, and $5{,}000$ iterations on NVIDIA A40 GPUs.
All models are trained only under AWGN; generalization to Rayleigh, Rician, and imperfect-CSI conditions is evaluated without retraining, providing a stringent test of communication robustness.

\subsection{Result Analysis}

\subsubsection{Connection-scheme comparison}
\label{sec:result_connection}

\begin{table*}[t]
\centering
\caption{WikiText-103 perplexity (out-of-distribution) under AWGN and Rayleigh fading at SNR $\in\{-15,\ldots,15\}$\,dB.
All methods use the entropy bottleneck with $\lambda{=}0.01$.
The ``w/o noise'' column denotes the clean baseline without channel. 6L and 48L denotes 6-layers and 48-layers, respectively. The value is the lower the better. Best result per depth group is \textbf{bolded}.}
\label{tab:wikitext_snr}
\setlength{\tabcolsep}{3.5pt}
\begin{tabular}{l r r r r r r r r}
\toprule
 & \multicolumn{8}{c}{WikiText-103 PPL ($\downarrow$)} \\
\cmidrule(lr){2-9}
 & \multicolumn{7}{c}{SNR (dB)} & \\
\cmidrule(lr){2-8}
Model & $-15$ & $-10$ & $-5$ & $0$ & $5$ & $10$ & $15$ & w/o noise \\
\midrule
\multicolumn{9}{c}{\textit{AWGN}} \\
\midrule
Baseline (EB)-6L  & $1.6{\times}10^5$ & \textbf{10\,110} & \textbf{3\,682} & 2\,620 & 2\,336 & 2\,245 & 2\,233 & 2\,221 \\
HC (EB)-6L        & \bm{$1.0{\times}10^6$} & 20\,270 & 4\,364 & 2\,555 & 2\,132 & 1\,994 & 1\,954 & 1\,969 \\
mHC (EB)-6L       & $3.3{\times}10^6$ & 27\,896 & 4\,378 & \textbf{2\,286} & \textbf{1\,850} & \textbf{1\,737} & \textbf{1\,687} & \textbf{1\,677} \\
\cmidrule(lr){1-9}
Baseline (EB)-48L & \textbf{28\,478} & \textbf{7\,805} & 4\,669 & 3\,947 & 3\,729 & 3\,623 & 3\,639 & 3\,587 \\
HC (EB)-48L       & $3.0{\times}10^6$ & 50\,767 & 9\,062 & 5\,771 & 5\,337 & 5\,235 & 5\,135 & 5\,125 \\
mHC (EB)-48L      & $2.5{\times}10^5$ & 12\,764 & \textbf{4\,139} & \textbf{2\,765} & \textbf{2\,444} & \textbf{2\,316} & \textbf{2\,290} & \textbf{2\,300} \\
\midrule
\multicolumn{9}{c}{\textit{Rayleigh fading}} \\
\midrule
Baseline (EB)-6L  & \bm{$2.9{\times}10^8$} & \bm{$2.8{\times}10^5$} & \textbf{15\,091} & \textbf{4\,690} & 2\,910 & 2\,419 & 2\,294 & 2\,221 \\
HC (EB)-6L        & $1.4{\times}10^{10}$ & $4.0{\times}10^6$ & 34\,653 & 7\,429 & 3\,015 & 2\,291 & 2\,064 & 1\,969 \\
mHC (EB)-6L       & $1.7{\times}10^{11}$ & $1.1{\times}10^8$ & 93\,608 & 8\,238 & \textbf{2\,861} & \textbf{2\,074} & \textbf{1\,785} & \textbf{1\,677} \\
\cmidrule(lr){1-9}
Baseline (EB)-48L & \bm{$1.2{\times}10^6$} & \textbf{49\,824} & \textbf{10\,392} & \textbf{5\,555} & 4\,345 & 3\,797 & 3\,701 & 3\,587 \\
HC (EB)-48L       & $5.0{\times}10^9$ & $4.0{\times}10^6$ & 66\,878 & 15\,367 & 7\,169 & 5\,717 & 5\,555 & 5\,125 \\
mHC (EB)-48L      & $6.6{\times}10^8$ & $4.9{\times}10^5$ & 24\,804 & 6\,194 & \textbf{3\,252} & \textbf{2\,618} & \textbf{2\,371} & \textbf{2\,300} \\
\bottomrule
\end{tabular}
\end{table*}

\paragraph{FineWeb10B perplexity (SemCom)}
Fig.~\ref{fig:snr_ppl} plots perplexity (log-scale) across SNR for the in-distribution FineWeb10B validation set.
Under AWGN channel, mHC-48L achieves the lowest PPL at SNR ${\ge}{-5}$\,dB, e.g., $685$ at $0$\,dB versus $996$ for HC-48L and $1{,}035$ for Baseline-48L ($34\%$ and $34\%$ reduction, respectively).
For the 6-layer group, mHC-6L and HC-6L are closely matched at moderate-to-high SNR (${\ge}0$\,dB), both substantially outperforming the Baseline; at $10$\,dB, mHC-6L reaches $427$ versus $435$ for HC-6L and $463$ for Baseline-6L.
Under Rayleigh fading, mHC-48L leads for SNR ${\ge}0$\,dB; at $10$\,dB it achieves $639$, compared to $917$ for HC-48L and $1{,}005$ for Baseline-48L, yielding $30\%$ and $36\%$ reductions, respectively.

\paragraph{WikiText-103 perplexity (SemCom)}
Table~\ref{tab:wikitext_snr} reports PPL on WikiText-103, an out-of-distribution corpus not seen during training.
At moderate-to-high SNR (${\ge}0$\,dB), mHC achieves the lowest PPL under both channels at both model scales.
Under AWGN at $10$\,dB, mHC-6L attains PPL $1{,}737$, which is $12.9\%$ lower than HC-6L ($1{,}994$) and $22.6\%$ lower than Baseline-6L ($2{,}245$).
For the 48-layer setting, the advantage is even more pronounced: mHC-48L reaches $2{,}316$ versus $5{,}235$ for HC-48L ($55.8\%$ reduction) and $3{,}623$ for Baseline-48L ($36.1\%$ reduction).

\paragraph{AG~News classification accuracy (TOC)}
Fig.~\ref{fig:agnews_snr} shows the 4-class AG~News accuracy across SNR under both channels.
Among the 6-layer variants, HC-6L achieves the highest accuracy, reaching $51.9\%$ (AWGN) and $49.6\%$ (Rayleigh) at $10$\,dB, with mHC-6L following at $47.6\%$/$37.5\%$; both significantly outperform Baseline-6L (${\sim}30\%$), confirming that multi-stream connectivity improves task-relevant feature extraction.
In the 48-layer setting, mHC-48L leads under both channels ($31.5\%$/$31.4\%$ at $10$\,dB), outperforming Baseline-48L ($25.1\%$/$25.5\%$) and HC-48L ($24.0\%$/$25.4\%$).

\begin{figure*}[t]
\centering
\includegraphics[width=0.8\textwidth]{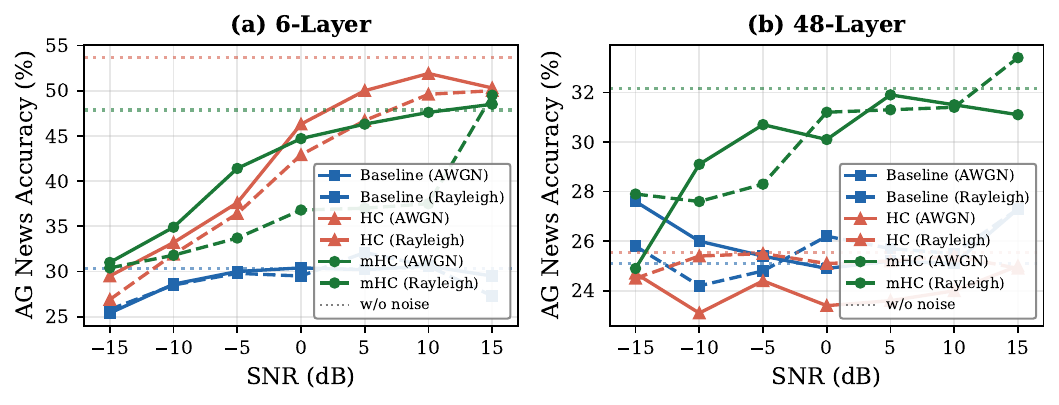}
\caption{AG~News 4-class accuracy vs.\ SNR under AWGN and Rayleigh fading for 6-layer and 48-layer models.
All methods use the entropy bottleneck with $\lambda{=}0.01$.
HC and mHC consistently outperform the single-stream Baseline, with mHC-48L leading at 48 layers.}
\label{fig:agnews_snr}
\end{figure*}

\subsubsection{Channel-coding component comparison}
\label{sec:result_channel_coding}

Table~\ref{tab:channel_coding_comparison} compares the 48-layer mHC~(EB) against Dense~(DeepSC) and VIB under Rayleigh fading channel.
Dense achieves the lowest clean PPL ($322$) because the unquantized encoder introduces no information loss in the absence of noise; its rate is marked ``--'' since it transmits continuous-valued symbols without quantization or entropy coding, making a bit-rate metric undefined.
However, without rate control, Dense degrades catastrophically under fading: at $0$\,dB its PPL explodes to $37{,}014$ ($115{\times}$ clean), whereas mHC~(EB) reaches only $1{,}456$ ($2.6{\times}$ clean).
mHC~(EB) outperforms Dense at all SNR levels below $10$\,dB, and even at $10$\,dB the gap is marginal ($639$ vs.\ $527$, $21\%$).
VIB again exhibits posterior collapse (effective rate $= 0$, PPL ${\sim}2{,}140$ regardless of SNR).
These results demonstrate that EB-based rate control is essential for reliable operation under realistic fading channels, where instantaneous SNR frequently drops below $5$\,dB.

\begin{table}[t]
\centering
\caption{Channel-coding component comparison (48-layer, Rayleigh fading). FineWeb10B PPL at selected SNR levels and effective rate.
mHC~(EB) uses $\lambda{=}0.01$; Dense has no rate control; VIB uses $\lambda{=}0.01$.
Best PPL \emph{at each SNR} is \textbf{bolded}.}
\label{tab:channel_coding_comparison}
\setlength{\tabcolsep}{4pt}
\begin{tabular}{l r r r r r}
\toprule
 & \multicolumn{4}{c}{FineWeb PPL ($\downarrow$)} & \\
\cmidrule(lr){2-5}
Model & Clean & $0$\,dB & $5$\,dB & $10$\,dB & Rate \\
\midrule
Dense (DeepSC)~\cite{xie2021deepsc} & \textbf{322}  & 37\,014         & 1\,478         & \textbf{527} & --   \\
VIB (IB)~\cite{shao2022ib_edge_inference}       & 2\,140        & 2\,150          & 2\,127         & 2\,127       & 0    \\
mHC (EB)       & 562           & \textbf{1\,456} & \textbf{790}   & 639          & 249K \\
\bottomrule
\end{tabular}
\end{table}

\subsubsection{Channel robustness under Rician fading and imperfect CSI}
\label{sec:exp_channel_robust}

To evaluate link-level robustness beyond the AWGN and Rayleigh conditions used above, we test the same AWGN-trained checkpoints (no retraining) under two additional wireless impairments.

\paragraph{Rician fading}
Table~\ref{tab:rician} reports FineWeb10B PPL under Rician fading with $K$-factors $K{\in}\{1,5,10\}$ at representative SNR levels for both model scales.
Recall that all models are trained exclusively under AWGN; these results therefore measure zero-shot generalization to unseen fading conditions.
As $K$ increases (stronger line-of-sight component), PPL approaches the clean reference, confirming the expected physical behavior.
Across all 24 (model, $K$, SNR) combinations, mHC consistently achieves the lowest PPL within each model scale, with an average improvement of $27\%$ over the baseline model in the 6-layer setting and $40\%$ in the 48-layer setting.
The advantage is the most pronounced under severe fading ($K{=}1$, low SNR), where the EB-compressed representation must contend with the largest channel uncertainty.

\begin{table}[t]
\centering
\caption{FineWeb10B PPL under Rician fading (AWGN-trained models, no retraining). $K$ denotes the Rician $K$-factor; $K{=}0$ corresponds to Rayleigh. Lower is better; best per condition is \textbf{bolded}.}
\label{tab:rician}
\setlength{\tabcolsep}{3pt}
\begin{tabular}{l l c r r r r r}
\toprule
 & & & \multicolumn{5}{c}{FineWeb10B PPL ($\downarrow$)} \\
\cmidrule(lr){4-8}
 & & & & \multicolumn{4}{c}{SNR (dB)} \\
\cmidrule(lr){5-8}
Scale & Model & $K$ & Clean & $-5$ & $0$ & $5$ & $10$ \\
\midrule
\multirow{6}{*}{6L}
 & Baseline & 1  & 560  & 621  & 591  & 589  & 581  \\
 & mHC      & 1  & \textbf{399}  & \textbf{475}  & \textbf{452}  & \textbf{438}  & \textbf{434}  \\
\cmidrule(lr){2-8}
 & Baseline & 5  & 560  & 594  & 573  & 573  & 567  \\
 & mHC      & 5  & \textbf{399}  & \textbf{455}  & \textbf{423}  & \textbf{417}  & \textbf{411}  \\
\cmidrule(lr){2-8}
 & Baseline & 10 & 560  & 595  & 568  & 571  & 557  \\
 & mHC      & 10 & \textbf{399}  & \textbf{444}  & \textbf{418}  & \textbf{410}  & \textbf{404}  \\
\midrule
\multirow{6}{*}{48L}
 & Baseline & 1  & 988  & 1051 & 1010 & 1033 & 1022 \\
 & mHC      & 1  & \textbf{567}  & \textbf{633}  & \textbf{611}  & \textbf{608}  & \textbf{612}  \\
\cmidrule(lr){2-8}
 & Baseline & 5  & 988  & 1024 & 1016 & 992  & 995  \\
 & mHC      & 5  & \textbf{567}  & \textbf{614}  & \textbf{589}  & \textbf{592}  & \textbf{576}  \\
\cmidrule(lr){2-8}
 & Baseline & 10 & 988  & 1035 & 999  & 991  & 989  \\
 & mHC      & 10 & \textbf{567}  & \textbf{607}  & \textbf{583}  & \textbf{568}  & \textbf{573}  \\
\bottomrule
\end{tabular}
\end{table}

\paragraph{Imperfect CSI}
We further evaluate robustness under imperfect receiver-side CSI.
The receiver observes an estimated channel coefficient $\hat{h} = h + e$, where $e \sim \mathcal{CN}(0, \sigma_e^2)$, and applies zero-forcing (ZF) equalization.
Table~\ref{tab:imperfect_csi} reports FineWeb10B PPL for the 6-layer models under mild estimation error ($\sigma_e^2{=}0.001$).
mHC achieves lower PPL than the baseline at $-5$, $0$, and $5$\,dB.
For example, at $5$\,dB, mHC achieves PPL~$457$, compared with PPL~$1{,}520$ for the baseline.
These results show that mHC remains robust under mild receiver-side CSI errors.

\begin{table}[t]
\centering
\caption{FineWeb10B PPL under imperfect CSI ($\sigma_e^2{=}0.001$) with Rayleigh fading and ZF equalization (6-layer, AWGN-trained models). Best performance is \textbf{bolded}.}
\label{tab:imperfect_csi}
\setlength{\tabcolsep}{4pt}
\begin{tabular}{l r r r r r}
\toprule
 & \multicolumn{5}{c}{FineWeb10B PPL ($\downarrow$)} \\
\cmidrule(lr){2-6}
 & & \multicolumn{4}{c}{SNR (dB)} \\
\cmidrule(lr){3-6}
Model & Clean & $-5$ & $0$ & $5$ & $10$ \\
\midrule
Baseline (EB) & 552  & 2\,099 & 655  & 1\,520 & \textbf{579}  \\
mHC (EB)      & \textbf{397}  & \textbf{813}   & \textbf{534}  & \textbf{457}   & 697  \\
\bottomrule
\end{tabular}
\end{table}

\subsubsection{Resource Usage}

Table~\ref{tab:resource_detail} reports the parameter count, per-component overhead, and wall-clock training time for all configurations.
The mHC mixing matrices account for only $0.049\%$ of the total parameters in the 6-layer model and $0.200\%$ in the 48-layer model (${\sim}21$K and ${\sim}91$K extra parameters, respectively), confirming the $\mathcal{O}(S^2/d^2)$ scaling predicted in Sec.~\ref{sec:method_resource}.
The channel encoder/decoder adds a fixed $102$K or $33$K parameters regardless of the connectivity scheme. Thus, mHC introduces no extra transmission-side cost.
In terms of FLOPs, the three $S{\times}S$ mixing operations per layer constitute approximately $4.2\%$ of the per-layer Transformer FLOPs for the 6-layer model ($S{=}4$, $d{=}288$) and ${\sim}8.0\%$ for the 48-layer model ($S{=}4$, $d{=}150$), consistent with the $\mathcal{O}(S/d)$ ratio in~\eqref{eq:flops_ratio}.
Notably, mHC-6L converges in $1$h$\,35$m, which is only $59\%$ of the HC-6L training time ($2$h$\,41$m) despite the additional Sinkhorn projection overhead; this speedup is attributable to the DS constraint stabilizing optimization and enabling faster convergence (cf.\ Sec.~\ref{sec:method_train}).
For the 48-layer models, mHC-48L and HC-48L exhibit comparable training times (${\sim}2$h$\,40$m), suggesting that the convergence benefit of the DS constraint is more pronounced at shallow depth.

\begin{table}[t]
\centering
\caption{Resource comparison across connection schemes.
HC Params denotes the number of parameters introduced by the multi-stream mixing matrices.
Channel Params is the combined channel encoder/decoder overhead.
Training time is measured on a single NVIDIA A40 GPU.}
\label{tab:resource_detail}
\setlength{\tabcolsep}{3.5pt}
\begin{tabular}{l
  S[table-format=2.2]
  S[table-format=5.0]
  S[table-format=1.3]
  S[table-format=3.0]
  c}
\toprule
{Model} & {Total (M)} & {HC Params} & {HC (\%)} & {Ch.\ Params (K)} & {Time} \\
\midrule
Baseline-6L  & 43.32 &      0 & 0.000 & 102 & 1h\,03m \\
HC-6L        & 43.34 & 21048 & 0.049 & 102 & 2h\,41m \\
mHC-6L       & 43.34 & 21336 & 0.049 & 102 & 1h\,35m \\
\midrule
Baseline-48L & 45.55 &      0 & 0.000 &  33 & 0h\,39m \\
HC-48L       & 45.66 & 88896 & 0.195 &  33 & 2h\,34m \\
mHC-48L      & 45.66 & 91200 & 0.200 &  33 & 2h\,46m \\
\bottomrule
\end{tabular}
\end{table}

\subsubsection{Training Stability and Convergence}

\begin{figure*}[t]
\centering
\includegraphics[width=0.8\textwidth]{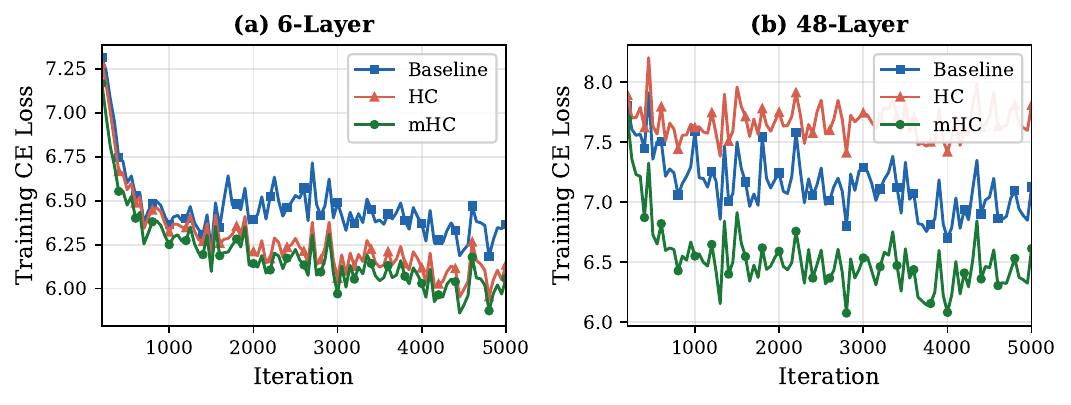}
\caption{Training convergence (CE loss vs.\ iteration) for 6-layer and 48-layer models.
mHC converges to the lowest loss among all methods at both depths, demonstrating the training stability benefit of DS-constrained mixing.}
\label{fig:training_convergence}
\end{figure*}

Fig.~\ref{fig:training_convergence} plots the training CE loss as a function of iterations for both depths.
mHC consistently reaches the lowest loss among all methods: for the 6-layer model, mHC-6L converges below both HC-6L and Baseline-6L throughout training; for the 48-layer model, the gap is even more pronounced, with HC-48L exhibiting visibly slower convergence while mHC-48L converges smoothly to a substantially lower loss.
Combined with the wall-clock times in Table~\ref{tab:resource_detail} (mHC-6L: $1$h$\,$35m vs.\ HC-6L: $2$h$\,$41m, a $41\%$ speedup), these results confirm that the DS constraint stabilizes optimization and accelerates convergence, supporting the analysis in Sec.~\ref{sec:method_resource}.

\subsubsection{Ablation Study}
\label{sec:ablation}

We ablate three key hyperparameters of the mHC coding scheme framework on the mHC-6L model under AWGN ($\mathrm{SNR}{=}10$\,dB).

\paragraph{Bottleneck dimension $k$}
Table~\ref{tab:ablation_combined} (a) reports the effect of varying the bottleneck dimension $k$, i.e., the output dimension of the channel encoder $g_\phi$ that determines the number of transmitted complex symbols.
The PPL shows a non-monotonic trend as $k$ increases.
Increasing $k$ from $16$ to $64$ reduces PPL from $554.2$ to $414.3$, since more channel symbols are available for semantic transmission.
However, when $k$ is further increased to $128$ and $256$, the PPL rises to $433.1$ and $575.5$, respectively.
This suggests that simply increasing the bottleneck dimension does not always improve the rate-distortion tradeoff.
A larger $k$ provides more channel symbols, but it also leads to more feature dimensions to be quantized and learned by the EB.
Also, the estimated bit rate increases almost linearly with $k$, showing that the transmission load grows with the bottleneck dimension.

\paragraph{Channel scale $\sigma_c$}
The channel scale $\sigma_c$ is a learnable value that sets the target standard deviation of the channel encoder output before EB quantization, thereby controlling the effective dynamic range of the quantized features.
Table~\ref{tab:ablation_combined} (b) investigates its initialization value.
A very small scale ($\sigma_c{=}1$) collapses the quantization range, yielding PPL~$2{,}505$.
Performance improves rapidly as $\sigma_c$ increases to $5$--$10$, with the optimum at $\sigma_c{=}10$ (PPL~$382.0$), and degrades slightly at $\sigma_c{=}20$ due to the coarser effective quantization step size.

\paragraph{Rate penalty $\lambda$}
Table~\ref{tab:ablation_combined} (c) sweeps the Lagrange multiplier $\lambda$ from $0.001$ to $0.1$.
As expected, a smaller $\lambda$ leads to less distortion: $\lambda{=}0.001$ achieves the lowest PPL ($266.7$) but provides little rate control.
Larger $\lambda$ encourages compression, increasing PPL while leaving the rate approximately constant (the EB saturates at the same quantization level).
The default $\lambda{=}0.01$ provides a reasonable balance between rate and distortion.

\begin{table*}[t]
\centering
\caption{Ablation studies on mHC-6L under AWGN ($\mathrm{SNR}{=}10$\,dB). Default settings are marked with $\dagger$.}
\label{tab:ablation_combined}
\setlength{\tabcolsep}{4pt}
\begin{minipage}[t]{0.30\textwidth}
\centering
\small
(a) Bottleneck dimension $k$\\[4pt]
\begin{tabular}{S[table-format=3.0] S[table-format=3.1] S[table-format=6.0]}
\toprule
{$k$} & {PPL} & {Rate} \\
\midrule
16  & 554.2  & 62208  \\
32  & 443.9  & 124416 \\
{64$^\dagger$}  & \textbf{414.3}  & 248832 \\
128 & 433.1  & 496783 \\
256 & 575.5  & 991232 \\
\bottomrule
\end{tabular}
\end{minipage}\hfill
\begin{minipage}[t]{0.30\textwidth}
\centering
\small
(b) Channel scale $\sigma_c$\\[4pt]
\begin{tabular}{S[table-format=2.0] S[table-format=4.1] S[table-format=6.0]}
\toprule
{$\sigma_c$} & {PPL} & {Rate} \\
\midrule
1   & 2504.7 & 149504 \\
2   & 452.0  & 151552 \\
{5$^\dagger$}   & 424.6  & 248832 \\
10  & \textbf{382.0}  & 335872 \\
20  & 404.9  & 409600 \\
\bottomrule
\end{tabular}
\end{minipage}\hfill
\begin{minipage}[t]{0.34\textwidth}
\centering
\small
(c) Rate penalty $\lambda$\\[4pt]
\begin{tabular}{S[table-format=1.3] S[table-format=3.1] S[table-format=6.0]}
\toprule
{$\lambda$} & {PPL} & {Rate} \\
\midrule
0.001 & \textbf{266.7}  & 248832 \\
0.005 & 375.6  & 248832 \\
{0.01$^\dagger$}  & 416.6  & 248832 \\
0.02  & 463.6  & 248832 \\
0.05  & 589.4  & 248832 \\
0.1   & 703.2  & 248525 \\
\bottomrule
\end{tabular}
\end{minipage}
\end{table*}
\section{Conclusions}
\label{sec:conclusion}

We propose an \textit{mHC coding scheme} for resource-efficient SemCom and TOC over wireless channels. The proposed scheme combines an mHC encoder with an EB, thereby improving semantic representation learning while enabling explicit entropy-coded rate control. We show that the DS constraint on stream mixing does not increase the ideal coding length. Thus, the multi-stream encoder introduces no additional channel uses after stream reduction. Numerical results under AWGN, Rayleigh fading, Rician fading, and imperfect CSI show the effectiveness of the proposed scheme. Compared with residual and unconstrained HC baselines, mHC achieved better perplexity, out-of-distribution generalization, task accuracy, and convergence stability with negligible parameters and computational overheads. The channel-coding comparison further shows that the EB provides practical RD control, whereas dense coding lacks an explicit rate measure, and VIB-based coding can collapse to a zero-rate solution. These results indicate that combining balanced multi-stream semantic encoding with EB rate control is a promising approach for robust and bandwidth-efficient learned wireless communication.

\printbibliography

\appendices

\section{Proof of Theorem~\ref{thm:lin_entropy_main}}
\label{app:proof_lin_entropy}
\begin{IEEEproof}
Let $\bm z=\vecop(\bZ)\in\RR^{Sd}$ and $\widetilde{\bm z}=\vecop(\widetilde{\bZ})$, where $\vecop(\cdot)$ stacks the columns of its argument.
Since $\widetilde{\bZ}=\bH\bZ$ applies the same $S\times S$ mixing to every column of $\bZ$, the standard Kronecker--vec identity $\vecop(\mathbf{A}\mathbf{B}\mathbf{C})=(\mathbf{C}^\top\!\otimes\mathbf{A})\vecop(\mathbf{B})$ with $\mathbf{A}=\bH$, $\mathbf{B}=\bZ$, $\mathbf{C}=\bI_d$ gives
\begin{equation}
\widetilde{\bm z} = (\bI_d\otimes \bH)\,\bm z.
\end{equation}
For an invertible linear transform $\widetilde{\bm z}=\bA\bm z$, differential entropy satisfies
$h(\widetilde{\bm z}) = h(\bm z) + \log|\det(\bA)|$.
Setting $\bA=\bI_d\otimes\bH$ and applying the mixed-product property of Kronecker determinant yields
\begin{equation}
\det(\bI_d\otimes\bH)=\det(\bI_d)^{S}\det(\bH)^{d}=\det(\bH)^d.
\end{equation}
Therefore,
\begin{equation}
h\!\big(\vecop(\widetilde{\bZ})\big)
= h\!\big(\vecop(\bZ)\big) + d\log|\det(\bH)|.
\end{equation}
\end{IEEEproof}

\section{Proof of Theorem~\ref{thm:det_bound_main}}
\label{app:proof_det_bound}
\begin{IEEEproof}
By Hadamard's inequality,
\begin{equation}
|\det(\bH)| \le \prod_{i=1}^{S} \|\bm r_i\|_2,
\end{equation}
where $\bm r_i^\top$ denotes the $i$-th row of $\bH$.
For $\bH\in\DS_S$, each row is elementwise nonneg\-ative and sums to one. Thus, we have $\|\bm r_i\|_2 \le \|\bm r_i\|_1 = 1$ for every $i$ and $|\det(\bH)|\le 1$.
For the equality case, a permutation matrix $\mathbf{P}\in\DS_S$ has orthonormal rows, giving $\|\bm r_i\|_2=1$ for all $i$ and $|\det(\mathbf{P})|=1$.
Conversely, equality in Hadamard's bound requires each row to satisfy $\|\bm r_i\|_2=\|\bm r_i\|_1=1$, which forces every row to be a standard basis vector; combined with the DS constraints, this implies $\bH$ is a permutation matrix.
\end{IEEEproof}

\section{Proof of Theorem~\ref{thm:rd_dominance_main}}
\label{app:proof_rd_dominance}
\begin{IEEEproof}
By Lemma~\ref{lem:baseline_inclusion_main}, $\mathcal{F}_{\mathrm{base}}\subseteq\mathcal{F}_{\mathrm{mHC}}$.
Consider any constraint $R\ge 0$ on the expected EB rate. 
Every encoder/decoder pair in $\mathcal{F}_{\mathrm{base}}$ whose expected EB rate is no larger than $R$ and whose distortion is $D$ is also contained in $\mathcal{F}_{\mathrm{mHC}}$ at the same rate and distortion.
Since the infimum over a superset cannot exceed the infimum over a subset,
\begin{equation}
D^\star_{\mathrm{mHC}}(R)
= \inf_{\substack{f\in\mathcal{F}_{\mathrm{mHC}}\\ \mathbb{E}[\mathrm{Rate}]\le R}} \mathbb{E}[\ell]
\le \inf_{\substack{f\in\mathcal{F}_{\mathrm{base}}\\ \mathbb{E}[\mathrm{Rate}]\le R}} \mathbb{E}[\ell]
= D^\star_{\mathrm{baseline}}(R).
\end{equation}
The equivalent statement $R^\star_{\mathrm{mHC}}(D) \le R^\star_{\mathrm{baseline}}(D)$ follows from the monotonicity of the RD function.
\end{IEEEproof}


\end{document}